\documentclass[journal]{IEEEtran}

\usepackage[english]{babel}
\usepackage{booktabs}

\usepackage{ifpdf}
\usepackage{bbm}

\usepackage{cite} 
\usepackage{url}
\usepackage{hyperref}

\ifCLASSINFOpdf
	\usepackage[pdftex]{graphicx}
	\graphicspath{{./}}
\else
	\usepackage[dvips]{graphicx}
	\graphicspath{./}
\fi
\usepackage{color}

\usepackage{pgf, tikz, pgfplots}
\usetikzlibrary{shapes, arrows, automata}
\usetikzlibrary{calc,hobby,decorations}

\usepackage[cmex10]{amsmath}
\usepackage{amsfonts, amssymb, amsthm}
\usepackage{mathrsfs}

\usepackage{algorithm} 
\usepackage{algorithmic}
\usepackage{enumerate}
\usepackage{multirow}
\usepackage{rotating}
\usepackage{subcaption}
\usepackage{needspace}

\input{mySymbol.sty}

\definecolor{penndarkestblue}{cmyk}{1,0.74,0,0.77}
\definecolor{penndarkerblue}{cmyk}{1,0.74,0,0.70}
\definecolor{pennblue}{cmyk}{0.99,0.66,0,0.57} 
\definecolor{pennlighterblue}{cmyk}{0.98,0.44,0,0.35}
\definecolor{pennlightestblue}{cmyk}{0.38,0.17,0,0.17} 

\definecolor{penndarkestred}{cmyk}{0,1,0.89,0.66}
\definecolor{penndarkerred}{cmyk}{0,1,0.88,0.55}
\definecolor{pennred}{cmyk}{0,1,0.83,0.42} 
\definecolor{pennlighterred}{cmyk}{0,1,0.6,0.24}
\definecolor{pennlightestred}{cmyk}{0,0.43,0.26,0.12} 

\definecolor{penndarkestgreen}{cmyk}{1,0,1,0.68}
\definecolor{penndarkergreen}{cmyk}{1,0,1,0.57}
\definecolor{penngreen}{cmyk}{1,0,1,0.44} 
\definecolor{pennlightergreen}{cmyk}{1,0,1,0.25}
\definecolor{pennlightestgreen}{cmyk}{0.43,0,0.43,0.13}

\definecolor{penndarkestorange}{cmyk}{0,0.65,1,0.49}
\definecolor{penndarkerorange}{cmyk}{0,0.65,1,0.33}
\definecolor{pennorange}{cmyk}{0,0.54,1,0.24} 
\definecolor{pennlighterorange}{cmyk}{0,0.32,1,0.13}
\definecolor{pennlightestorange}{cmyk}{0,0.15,0.46,0.06}
	
\definecolor{penndarkestpurple}{cmyk}{0,1,0.11,0.86}
\definecolor{penndarkerpurple}{cmyk}{0,1,0.13,0.82}
\definecolor{pennpurple}{cmyk}{0,1,0.11,0.71} 
\definecolor{pennlighterpurple}{cmyk}{0,1,0.05,0.46}
\definecolor{pennlightestpurple}{cmyk}{0,0.35,0.02,0.23}
	
\definecolor{pennyellow}{cmyk}{0,0.20,1,0.05} 
\definecolor{pennlightgray1}{cmyk}{0,0,0,0.05}
\definecolor{pennlightgray3}{cmyk}{0.01,0.01,0,0.18}
\definecolor{pennmediumgray1}{cmyk}{0.04,0.03,0,0.31}
\definecolor{pennmediumgray4}{cmyk}{0.08,0.06,0,0.54}
\definecolor{penndarkgray2}{cmyk}{0.09,0.07,0,0.71}
\definecolor{penndarkgray4}{cmyk}{0.1,0.1,0,0.92}

\def\SO3{\mathrm{SO(3)}}

\newtheorem{assumption}{\hspace{0pt}\bf Assumption \hspace{-0.15cm}}

\newtheorem{theorem}{\hspace{0pt}\bf Theorem}
\newtheorem{corollary}{\hspace{0pt}\bf Corollary}

\newtheorem{remark}{\hspace{0pt}\bf Remark}

\newtheorem{definition}{\hspace{0pt}\bf Definition}
\newtheorem{problem}{\hspace{0pt}\bf Problem}

\begin{document}

\title{Constrained Environment Optimization for Prioritized Multi-Agent Navigation}

\author{\IEEEauthorblockN{Zhan Gao and Amanda Prorok}\\
\thanks{This work is an extension of \cite{gao2022environment}. Department of Computer Science and Technology, University of Cambridge, Cambridge, UK (Email: zg292@cam.ac.uk, asp45@cam.ac.uk). This work was supported by European Research Council (ERC) Project 949940 (gAIa).}}

\markboth{}%
{Constrained Environment Optimization for Prioritized Multi-Agent Navigation}

\maketitle

\begin{abstract}
Traditional approaches to the design of multi-agent navigation algorithms consider the environment as a fixed constraint, despite the influence of spatial constraints on agents' performance. Yet hand-designing conducive environment layouts is inefficient and potentially expensive. The goal of this paper is to consider the environment as a decision variable in a system-level optimization problem, where both agent performance and environment cost are incorporated. Towards this end, we propose novel problems of \emph{unprioritized} and \emph{prioritized environment optimization}, where the former considers agents unbiasedly and the latter accounts for agent priorities. We show, through formal proofs, under which conditions the environment can change while guaranteeing completeness (i.e., all agents reach goals), and analyze the role of agent priorities in the environment optimization. We proceed to impose real-world constraints on the environment optimization and formulate it mathematically as a constrained stochastic optimization problem. Since the relation between agents, environment and performance is challenging to model, we leverage reinforcement learning to develop a model-free solution and a primal-dual mechanism to handle constraints. Distinct information processing architectures are integrated for various implementation scenarios, including online/offline optimization and discrete/continuous environment. Numerical results corroborate the theory and demonstrate the validity and adaptability of our approach.
\end{abstract}


\IEEEpeerreviewmaketitle

\section{Introduction} \label{sec:intro}

Multi-agent systems present an attractive solution to spatially distributed tasks, wherein motion planning among moving agents and obstacles is one of the central problems. To date, the primal focus in multi-agent motion planning has been on developing effective, safe, and near-optimal navigation algorithms \cite{silver2005cooperative, van2008reciprocal,van2005prioritized,desaraju2012decentralized, standley2011complete, gao2023learning}. These algorithms consider the \textit{agents' environment as a fixed constraint}, where structures and obstacles must be circumnavigated. In this process, mobile agents engage in negotiations with one another for right-of-way, driven by local incentives to minimize individual delays. However, environmental constraints may result in dead-locks, live-locks and prioritization conflicts, even for state-of-the-art algorithms \cite{mani2010search}. These insights highlight the impact of the environment on multi-agent navigation.

Not all environments elicit the same kinds of agent behaviors and individual navigation algorithms are susceptible to environmental artefacts; undesirable environments can lead to irresolution in path planning \cite{ruderman_Uncovering_2018}. To deal with such bottlenecks, spatial structures (e.g., intersections, roundabouts) and markings (e.g., lanes) are developed to facilitate path de-confliction \cite{boudet2021collections} but these concepts are based on legacy mobility paradigms, which ignore inter-agent communication, cooperation, and systems-level optimization. While it is possible to deal with the circumvention of dead-locks and live-locks through hand-designed environment templates, such hand-designing process is inefficient~\cite{cap_Prioritized_2015}.

Reconfigurable and automated environments are emerging as a new trend~\cite{wang2010new, bier2014robotic, custodio2020flexible, gao2022environment}, incorporating mechatronic devices to establish interactive relations with agents. Such environments have wide applications in next generation buildings, especially in structured settings where agents are expected to solve repetitive tasks (like warehouses, factories, restaurants, etc.). In tandem with that enabling technology, the goal of this paper is to consider the environment as a \emph{variable} in pursuit of the agents' incentives. We propose the problem of systematically \textit{optimizing an environment to improve the navigation performance of a given multi-agent system}. This allows us to facilitate multi-agent navigation, even if the agents' actions cannot be directly controlled, e.g., when agents are human with inherent preferences and unknown decision-making processes, or physical robots with hardware and computation limitations. More in detail, our contributions are as follows:
\begin{enumerate}[(i)]
	
	\item We define novel problems of \textit{unprioritized} and \textit{prioritized} environment optimization, where the former considers agents unbiasedly and the latter accounts for agent priorities. We develop two solution variants, i.e., \textit{offline} and \textit{online} environment optimization, that adapt to different implementation scenarios. 
	
	\item We conduct the completeness analysis for multi-agent navigation with environment optimization, which identifies the conditions under which all agents are guaranteed to reach their goals. We also analyze the effects of \textit{agent priorities} on environment optimization, and show how environment ``resources'' can be negotitated to improve performance in an ordered scheme.
	
	\item We impose practical constraints on the environment optimization, and formulate it as a constrained stochastic optimization problem. We leverage reinforcement learning and a primal-dual mechanism to develop a model-free method, which allows us to integrate different information processing architectures (e.g., CNNs, GNNs) as a function of the problem setting.
	
	\item We evaluate the proposed framework in both discrete and continuous environment settings. The results corroborate theoretical findings and show the proposed approach can generalize beyond training instances, adapts to various optimization objectives and constraints, and allows for decentralized implementation. 
	
\end{enumerate}

\noindent \textbf{Related work.} The problem of environment optimization is reminiscent of robot co-design \cite{tanaka2015sdp,tatikonda2004control,tzoumas2018sensing}, wherein sensing, computation, and control are jointly optimized. While this domain of robotics is \textit{robot-centric} (i.e., does not consider the environment as an optimization criteria), there are a few works that, similarly to our approach, use reinforcement learning to co-optimize objectives \cite{lipson2000automatic,hornby2003generative,cheney2018scalable}. A more closely related idea is exploited in \cite{ruderman_Uncovering_2018}, wherein the environment is adversarially optimized to fail the navigation tasks of state-of-the-art solutions. It conducts a worst-case analysis to shed light on directions in which more robust systems can be developed. On the contrary, we optimize the environment to \textit{facilitate} multi-agent navigation tasks.

The role of the environment on motion planning has been previously explored. Specifically, \cite{bennewitz2002finding, jager2001decentralized} 
emphasize the existence of congestion and deadlocks in undesirable environments and develop trajectory planning methods to escape from potential deadlocks. The work in~\cite{vcap2015complete} defines the concept of ``well-formed'' environments, in which the navigation tasks of all agents can be carried out successfully without collisions. In~\cite{wu2020multi}, Wu et al. show that the shape of the environment leads to distinct \emph{path prospects} for different agents, and that this information can be shared among the agents to optimize and coordinate their mobility. Gur et al. in \cite{gur2021adversarial} generate adversarial environments and account for the latter information to develop resilient navigation algorithms. However, none of these works consider optimizing the environment to improve system-wide navigation performance. 

Our problem is also related to the problem of environment design. The work in \cite{zhang2009general} design environments to influence a agent's decision and \cite{keren2015goal, kulkarni2020designing} focus on boosting the interpretability of a robot's behavior for human, while both are formulated for a single-agent scenario. The works of Hauser \cite{hauser2013minimum, hauser2014minimum} consider removing obstacles from the environment to improve navigation performance of one agent, while \cite{bellusci2020multi} extends a similar idea to multi-agent settings but limited for discrete environments. However, it may not be practical to remove obstacles in real-world applications.

\section{Problem Formulation} \label{sec:problem}

Let $\ccalE$ be a $2$-D environment described by a starting region $\ccalS$, a destination region $\ccalD$ and an obstacle region $\Delta$ without overlap, i.e., $\ccalS \bigcap \ccalD = \ccalS \bigcap \Delta = \ccalD \bigcap \Delta = \emptyset$. Consider a multi-agent system with $n$ agents $\ccalA = \{A_i\}_{i=1}^n$ distributed in $\ccalE$. The agent bodies are contained within circles of 
radii $\{r_i\}_{i=1}^n$. Agents are initialized at positions $\bbS = [\bbs_1,\ldots,\bbs_n]$ in $\ccalS$ and deploy a given trajectory planner $\pi_a$ to move towards destinations $\bbD = [\bbd_1,\ldots,\bbd_n]$ in $\ccalD$. Let $\bbrho = [\rho_1,\ldots,\rho_n]^\top$ be a set of priorities that represent the importance of agents w.r.t. the navigation tasks, i.e., the larger the priority, the more important the agent; in other words, for $\rho_1 \ge \cdots \ge \rho_n$, agent $A_1$ has the highest priority. We consider a decentralized trajectory planner $\pi_a$, i.e., $\pi_a$ can be executed locally with 
only neighborhood information.

The existing literature focuses on developing novel trajectory planners to improve navigation performance. However, the latter depends not only on the implemented planner but also on the surrounding environment. A ``well-formed'' environment with an appropriate obstacle region yields good performance for a simple planner, while a ``poorly-formed'' environment may result in poor performance for an advanced planner. This insight implies an important role played by the environment in multi-agent navigation, and motivates the definition of the problem of environment optimization. We consider two variants: (1) 
\emph{unprioritized environment optimization} and (2) \emph{prioritized environment optimization}.
\begin{problem}[Unprioritized Environment Optimization]\label{def:environmenProblem}
	\emph{Given an environment with an initial obstacle region and a multi-agent system of $n$ agents that are required to reach $n$ destinations in a labeled navigation problem, find a policy that optimizes the environment layout to improve the agents' path efficiency while guaranteeing that they reach their destinations without collision.}
\end{problem}
\begin{problem}[Prioritized Environment Optimization]\label{def:prioritizedEO}
	\emph{Given an environment with an initial obstacle region, a multi-agent system of $n$ agents $\{A_i\}_{i=1}^n$ that are required to reach $n$ labeled destinations, agent priorities $\bbrho$ that are ordered by the index $\rho_1\ge \rho_2 \ge \cdots\ge\rho_n$, and a metric $M(\cdot) : A \to \mathbb{R}$ that measures the navigation performance of agents,\footnote{Without loss of generality, we assume the lower the value of $M(\cdot)$, the better the performance, e.g., corresponding to traveled distance and navigation time.} 
		find a policy that optimizes the environment layout to guarantee that all agents reach their destinations without collision 
		and to improve agents' navigation performance according to their priorities, i.e.,
		\begin{align}\label{eq:boundedPerformance}
			M(A_i) \le C_i~\forall~i=1,\ldots,n
		\end{align} 
		and $C_1 \le C_2 \le \cdots \le C_n$ is ordered by the priorities $\bbrho$.
	}	
\end{problem}

\noindent Problem \ref{def:environmenProblem} considers agents with equal priority and optimizes the environment to improve their performance unbiasedly. Problem \ref{def:prioritizedEO} is an extension of Problem \ref{def:environmenProblem} that accounts for agent priorities. It not only guarantees the success of all navigation tasks but also upper-bounds the worst-case performance of agents by constants $\{C_i\}_{i=1}^n$ [cf. \eqref{eq:boundedPerformance}], which are ordered based on agent priorities $\bbrho$.

In Sec. \ref{sec:Completeness}, we analyze the completeness of multi-agent navigation in Problem \ref{def:environmenProblem} to show the effectiveness of environment optimization. In Sec. \ref{sec:prioritizedEO}, we conduct the completeness analysis for Problem \ref{def:prioritizedEO} and characterize the effects of agent priorities on environment optimization. In Sec. \ref{sec:constrainedEO}, we impose real-world constraints on the environment optimization and formulate the problem mathematically as a constrained stochastic optimization problem. We transform the latter into the dual domain to tackle constraints and combine reinforcement learning with a primal-dual method for solutions. Lastly, we evaluate the proposed approach numerically and 
corroborate theoretical findings in Sec. \ref{sec:experiments}.

\section{Completeness of Unprioritized System}\label{sec:Completeness}

In this section, we provide the completeness analysis of multi-agent navigation for unprioritized environment optimization [Problem \ref{def:environmenProblem}]. Specifically, the multi-agent system may fail to find collision-free trajectories in an environment with an unsatisfactory obstacle region. Environment optimization overcomes this issue by modifying the obstacle region to guarantee the navigation success for all agents. The principle of environment optimization is that the obstacle region $\Delta$ can be controlled (i.e., changed) while its area $|\Delta|$ remains the same, i.e., $|\Delta| = |\hat{\Delta}|$ where $\hat{\Delta}$ is the changed obstacle region. That is, we do not remove or add any obstacles from our environment. In the following, we consider both \emph{offline} environment optimization and \emph{online} environment optimization. 

\subsection{Offline Environment Optimization}

Offline environment optimization optimizes the obstacle region $\Delta$ based on the starting region $\ccalS$ and the destination region $\ccalD$, and completes the optimization before the agents start moving towards destinations. The optimized environment remains static during agent movement. The goal is to find an optimal obstacle region $\Delta^*$ that maximizes the navigation performance of the multi-agent system.

We introduce our notation as follows. Let $\partial \ccalS$, $\partial \ccalD$, $\partial\Delta$ be the boundaries of the regions $\ccalS$, $\ccalD$, $\Delta$ and $B(\bbs,r)$ be a closed disk centered at position $\bbs$ with a radius $r$. This allows us to represent each agent $A_i$ as the disk $B(\bbs_i, r_i)$ where $\bbs_i$ is the central position of the agent. Define $d(\bbs_i, \bbs_j)$ as the closest distance between agents $A_i$, $A_j$, i.e., $d(\bbs_i, \bbs_j) = \min \|\bbz_i - \bbz_j\|_2$ for any $\bbz_i \in B(\bbs_i, r_i)$ and $\bbz_j \in B(\bbs_j, r_j)$, $d(\bbs_i, \partial \ccalS)$ the closest distance between agent $A_i$ and the starting region boundary $\partial \ccalS$, i.e., $d(\bbs_i, \partial \ccalS) = \min \|\bbz_i - \bbz_\ccalS\|_2$ for any $\bbz_i \in B(\bbs_i, r_i)$ and $\bbz_\ccalS \in \partial \ccalS$, and $d(\partial \ccalS, \partial \Delta)$ the closest distance between the starting boundary $\partial \ccalS$ and the obstacle boundary $\partial \Delta$, i.e., $d(\partial \ccalS, \partial \Delta) = \min \|\bbz_\ccalS - \bbz_\Delta\|_2$ for any $\bbz_\ccalS \in  \partial \ccalS$ and $\bbz_\Delta \in \partial \Delta$. Analogous definitions apply for the destination boundary $\partial\ccalD$ and the obstacle boundary $\partial\Delta$. Let $\bbp(t): [0,\infty) \to \mathbb{R}^2$ be the trajectory representing the central position movement of an agent. The trajectory $\bbp_i$ of agent $A_i$ is collision-free w.r.t. the obstacle region $\Delta$ if $d\big(\bbp_i(t), \partial \Delta \big) \ge r_i$ for all $t\ge 0$. The trajectories $\bbp_i$, $\bbp_j$ of two agents $A_i$, $A_j$ are collision-free with each other if $\|\bbp_i(t) - \bbp_j(t)\|_2 \ge r_i + r_j$ for all $t \ge 0$. To conduct the completeness analysis, we need the following assumptions. 
\begin{assumption}\label{as:initialDistribution}
	The initial positions $\{\bbs_i\}_{i=1}^n$ in $\ccalS$ and the destinations $\{\bbd_i\}_{i=1}^n$ in $\ccalD$ 
	are distributed in a way such that the distance between either two agents or the agent and the region boundary is larger than the maximal agent size, i.e., $d(\bbs_i,\bbs_j) \ge 2 \hat{r}$, $d(\bbs_i,\partial \ccalS) \ge 2 \hat{r}$ and $d(\bbd_i,\bbd_j) \ge 2 \hat{r}$, $d(\bbd_i,\partial \ccalD) \ge 2 \hat{r}$ for $i,j\!=\!1,...,n$ with $\hat{r} \!=\! \max_{i=1,...,n} r_i$. 
\end{assumption}
\begin{assumption}\label{as:trajectory}
	The agents $\ccalA$ can compute optimal trajectories from the starting positions $\{\bbs_i\}_{i=1}^n$ to the destinations $\{\bbd_i\}_{i=1}^n$ in the presence of the obstacle region $\Delta$, and any trajectories generated for agents, will be followed precisely.
\end{assumption}

Assumption \ref{as:initialDistribution} indicates that the starting positions $\{\bbs_i\}_{i=1}^n$ (or goal positions $\{\bbd_i\}_{i=1}^n$) are not too close to each other, which is commonly satisfied in real-world navigation tasks. Assumption \ref{as:trajectory} is standard in completeness analysis for multi-agent navigation \cite{vcap2015complete}. With these preliminaries, we show the completeness with the offline environment optimization. \vspace{-4mm}

\begin{theorem}\label{thm:offlineCompleteness}
	\emph{Consider the multi-agent system in the environment $\ccalE$ with starting, destination and obstacle regions $\ccalS$, $\ccalD$ and $\Delta$ satisfying Assumptions \ref{as:initialDistribution}-\ref{as:trajectory}. Let $d_{\max}$ be the maximal distance between $\ccalS$ and $\ccalD$, i.e., $d_{\max} = \max_{\bbz_\ccalS,\bbz_\ccalD} \| \bbz_\ccalS - \bbz_\ccalD \|_2$ for any $\bbz_\ccalS \in \ccalS,\bbz_\ccalD \in \ccalD$. Then, if $\ccalE$ satisfies
		\begin{align}\label{eq:offlineCompletenessCondition}
			|\ccalE \setminus (\Delta \cup \ccalS \cup \ccalD)| \ge 2 n d_{\max} \hat{r}
		\end{align}
		where $\hat{r} = \max_{i=1,\ldots,n} r_i$ is the maximal radius of $n$ agents and $|\cdot|$ represents the region area, the offline environment optimization guarantees that the navigation tasks of all agents can be carried out successfully without collision.} 
\end{theorem} 
\begin{proof}
	See Appendix \ref{Appendix:Theorem1}.
\end{proof}

Theorem \ref{thm:offlineCompleteness} states that the offline environment optimization guarantees the success of all navigation tasks under a mild condition. It does not require any initial ``well-formed" environment but only an obstacle-free area of minimal size [cf. \eqref{eq:offlineCompletenessCondition}], which is common in real-world scenarios. The offline environment optimization depends on the starting region $\ccalS$ and the destination region $\ccalD$, and completes optimizing the obstacle region \emph{before} the agents start to move. This requires a computational overhead before each new navigation task. Moreover, we are interested in generalizing the problem s.t. $\ccalS$ and $\ccalD$ are allowed to be time-varying during deployments.

\subsection{Online Environment Optimization} 

Online environment optimization changes the obstacle region $\Delta$ based on instantaneous states of the agents, e.g., positions, velocities and accelerations, after the agents start moving towards destinations. In other words, the environment is being optimized concurrently with agent movement. The goal is to find the optimal obstacle policy $\pi_o$ that changes the obstacle region to maximize the navigation performance of the multi-agent system.

Specifically, define the starting region as the union of the starting positions $\ccalS = \bigcup_{i=1,...,n}B(\bbs_i, r_i)$ and the destination region as that of the destinations $\ccalD = \bigcup_{i=1,...,n}B(\bbd_i, r_i)$ s.t. $\ccalS \bigcap \ccalD = \emptyset$. Since the obstacle region $\Delta$ now changes continuously, we define the \emph{capacity} of the online environment optimization as the maximal changing rate of the obstacle area $\dot{\Delta}$, i.e., the maximal obstacle area that can be changed per time step. To proceed, we require 
an assumption for the environment with an empty obstacle region $\Delta = \emptyset$. 
\begin{assumption}\label{as:emptyEnvironment}
	For an environment $\ccalE$ with no obstacle region $\Delta = \emptyset$, all navigation tasks can be carried out successfully without collision and the corresponding agent trajectories $\{\bbp_i\}_{i=1}^n$ are non-overlap. 
\end{assumption} 

Assumption \ref{as:emptyEnvironment} is mild because an environment with no obstacle region is the best scenario for the multi-agent navigation. For an environment with a non-empty obstacle region $\Delta$, the following theorem shows the completeness with the online environment optimization. 
\begin{theorem}\label{thm:onlineCompleteness}
	\emph{Consider the multi-agent system of $n$ agents $\{A_i\}_{i=1}^n$ satisfying Assumption \ref{as:emptyEnvironment} with $n$ collision-free trajectories $\{\bbp_i(t)\}_{i=1}^n$ for the environment with an empty obstacle region. Let $\{\bbv_i(t)\}_{i=1}^n$ be the velocities along $\{\bbp_i(t)\}_{i=1}^n$, respectively. For the environment $\ccalE$ with a non-empty obstacle region $\Delta \subset \ccalE \setminus (\ccalS \cup \ccalD)$ and $\ccalE \setminus (\Delta \cup \ccalS \cup \ccalD) \ne \emptyset$, if the capacity of the online environment optimization satisfies
		\begin{align}\label{eq:onlineCompletenessCondition}
			\dot{\Delta} \ge 2 n \hat{r} \| \hat{\bbv} \|_2
		\end{align}
		where $\|\hat{\bbv}\|_2 = \max_{t \in [0,T]}\max_{i=1,...,n}\|\bbv_i(t)\|_2$ is the maximal norm of the velocities and $\hat{r} = \max_{i=1,\ldots,n} r_i$ is the maximal agent radius, the navigation tasks of all agents can be carried out successfully without collision.} 
\end{theorem} 
\begin{proof}
	See Appendix \ref{proof:Thm2}.
\end{proof}

Theorem \ref{thm:onlineCompleteness} states that the online environment optimization guarantees the success of all navigation tasks as well as its offline counterpart. The result is established under a mild condition on the changing rate of the obstacle region [cf. \eqref{eq:onlineCompletenessCondition}] rather than the initial obstacle-free area [cf. \eqref{eq:offlineCompletenessCondition}]. This is due to the fact that the online environment optimization changes the obstacle region concurrently with agent movement. Hence, it improves navigation performance only if timely actions can be taken based on instantaneous system states. The completeness analysis in Theorems \ref{thm:offlineCompleteness}-\ref{thm:onlineCompleteness} demonstrates theoretically the effectiveness of the proposed environment optimization, in improving the performance of the multi-agent navigation.

\section{Completeness of Prioritized System}\label{sec:prioritizedEO}

Unprioritized environment optimization guarantees the success of navigation tasks in scenarios with sufficient ``resources'', i.e., a sufficiently large obstacle-free area [Thm. \ref{thm:offlineCompleteness}] and a sufficiently large obstacle changing rate [Thm. \ref{thm:onlineCompleteness}]. However, this may not be the case for scenarios with reduced resources. In the latter circumstances, the environment needs to be optimized with respect to the conflicts of interest among different agents, and allocates resources according to \textit{priorities} in order to negotiate these conflicts. 

With the formulation of prioritized environment optimization [Problem \ref{def:prioritizedEO}], we overcome this issue by incorporating agent priorities into environment optimization. That is, we put more emphasis on agents with higher priorities to guide the negotiation. In the sequel, we show that prioritized environment optimization is capable of guaranteeing the success of all navigation tasks with reduced resources by sacrificing the navigation performance of agents with lower priorities. Analogous to Section \ref{sec:Completeness}, we analyze the completeness of offline and online prioritized environment optimization, respectively, which characterizes the explicit effects of agent priorities on the navigation performance.

\subsection{Offline Prioritized Environment Optimization} 

Offline environment optimization optimizes the obstacle region $\Delta$ before 
navigation and guarantees the success of all navigation tasks if $|\ccalE \setminus (\Delta \cup \ccalS \cup \ccalD)| \ge 2 n d_{\max} \hat{r}$ [Thm. \ref{thm:offlineCompleteness}]. We consider the reduced-resource scenario, where the initial obstacle-free area is smaller than $2 n d_{\max} \hat{r}$, i.e., $|\ccalE \setminus (\Delta \cup \ccalS \cup \ccalD)| < 2 n d_{\max} \hat{r}$. In this circumstance, offline prioritized environment optimization is able to maintain the success of all navigation tasks, at the cost of lower-priority agent performance. 

\begin{figure}
	\centering
	\includegraphics[width=0.7\linewidth]{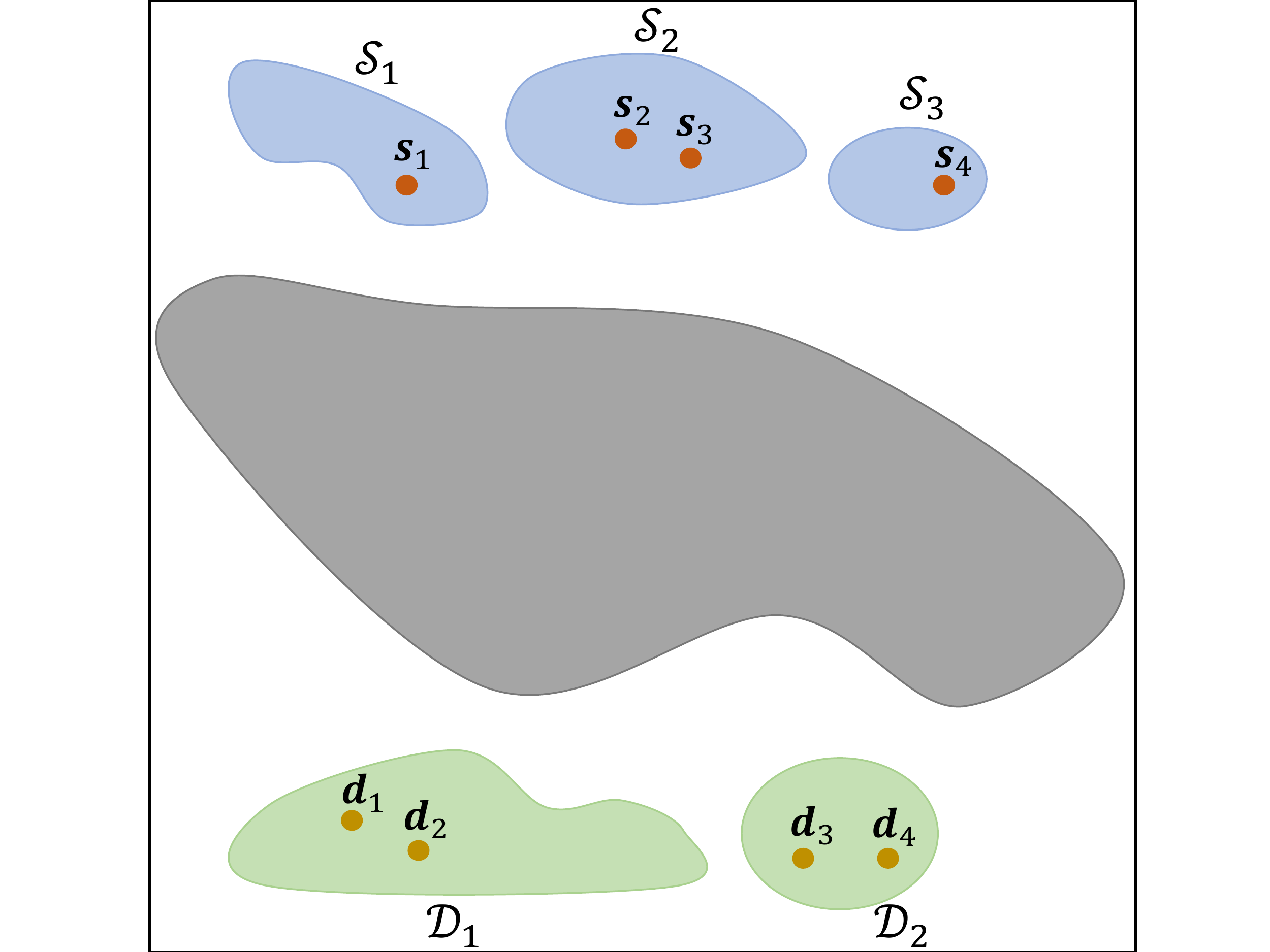}
	\caption{Example environment with $4$ starting positions $\{\bbs_i\}_{i=1}^4$ and $4$ destinations $\{\bbd_i\}_{i=1}^4$. There are $C_{\ccalS}=3$ self-connected starting components $\ccalS_1, \ccalS_2, \ccalS_3$ in the starting region $\ccalS$ (blue), $C_{\ccalD}=2$ self-connected destination components $\ccalD_1, \ccalD_2$ in the destination region $\ccalD$ (green), and one obstacle region $\Delta$ (grey). Agent $A_1$ is initialized in the starting component $\ccalS_{A_1} = \ccalS_1$ and moves towards its goal in the destination component $\ccalD_{A_1} = \ccalD_1$. Analogous notation applies for the other agents $A_2$, $A_3$ and $A_4$.}\label{fig:environmentExample1}\vspace{-4mm}
\end{figure}

We introduce additional notation as follows. Define the starting region $\ccalS$ as \emph{disconnected} if there exist two points in $\ccalS$ that cannot be connected by a path inside $\ccalS$, and \emph{self-connected} otherwise. For a disconnected $\ccalS$, let $\{\ccalS_{l}\}_{l=1}^{C_\ccalS}$ be the least number of self-connected components in $\ccalS$ s.t. their union covers $\ccalS$, i.e., $\bigcup_{l=1}^{C_\ccalS} \ccalS_l = \ccalS$.\footnote{For a self-connected $\ccalS$, we have $C_\ccalS = 1$ and $\ccalS_1 = \ccalS$.} 
Assume that agent $A_i$ is initialized in one of these component $\ccalS_{A_i}\in \{\ccalS_{l}\}_{l=1}^{C_\ccalS}$ for $i=1,\ldots,n$. Analogous definitions apply for the destination region $\ccalD$ -- see Fig. \ref{fig:environmentExample1} for an example environment. Let $d_{\max}(\ccalS_{l_1}, \ccalS_{l_2})$ be the maximal distance between the components $\ccalS_{l_1}$ and $\ccalS_{l_2}$ for $l_1 \ne l_2\in\{1,\ldots,C_\ccalS\}$, i.e., $d_{\max}(\ccalS_{l_1}, \ccalS_{l_2}) = \max \|\bbz_{\ccalS_{l_1}} - \bbz_{\ccalS_{l_2}}\|_2$ for any $\bbz_{\ccalS_{l_1}} \in \ccalS_{l_1}$ and  $\bbz_{\ccalS_{l_2}} \in \ccalS_{l_2}$. Denote by $\bbD_\ccalS \in \mathbb{R}^{C_\ccalS \times C_\ccalS}$ the matrix that collects the distances between the components $\{\ccalS_{l}\}_{l=1}^{C_\ccalS}$ with the $(l_1,l_2)$th entry $[\bbD_{\ccalS}]_{l_1 l_2} = d_{\max}(\ccalS_{l_1}, \ccalS_{l_2})$ for $l_1,l_2=1,\ldots,C_\ccalS$ and $[\bbD_{\ccalS}]_{ll} = 0$ for $l=1,\ldots,C_\ccalS$, $\bbD_\ccalD \in \mathbb{R}^{C_\ccalD \times C_\ccalD}$ the matrix that collects the distances between the components $\{\ccalD_k\}_{k=1}^{C_\ccalD}$, and $\bbD_{\ccalS,\ccalD} \in \mathbb{R}^{C_\ccalS \times C_\ccalD}$ the matrix that collects the distances between $\{\ccalS_{l}\}_{l=1}^{C_\ccalS}$ and $\{\ccalD_k\}_{k=1}^{C_\ccalD}$. For completeness analysis, we assume the following.
\begin{assumption}\label{as:distance}
	The distance between the starting components $\{\ccalS_{l}\}_{l=1}^{C_\ccalS}$ (or destination components $\{\ccalD_k\}_{k=1}^{C_\ccalD}$) is significantly smaller than the distance between the starting components $\{\ccalS_{l}\}_{l=1}^{C_\ccalS}$ and the destination components $\{\ccalD_k\}_{k=1}^{C_\ccalD}$, i.e., for any entry $d_{\ccalS}$ of $\bbD_{\ccalS}$, entry $d_{\ccalD}$ of $\bbD_{\ccalD}$ and entry $d_{\ccalS,\ccalD}$ of $\bbD_{\ccalS,\ccalD}$, it holds that $d_{\ccalS} + d_{\ccalD} \le d_{\ccalS,\ccalD}$.
\end{assumption}

Assumption \ref{as:distance} is reasonable because the starting positions are typically close to each other but far away from the destinations in real-world applications. With these in place, we characterize the completeness with the offline prioritized environment optimization in the following theorem.

\begin{theorem}\label{thm:prioritizedOfflineCompleteness}
	\emph{Consider the multi-agent system of $n$ agents $\{A_i\}_{i=1}^n$ in the environment $\ccalE$ with the same settings as Theorem \ref{thm:offlineCompleteness} and satisfying Assumption \ref{as:distance}. Let $\bbrho$ be the agent priorities ordered by the index $\rho_1\ge \rho_2 \ge \cdots \ge \rho_n$, agent $A_i$ be in the starting component $\ccalS_{A_i} \in \{\ccalS_{l}\}_{l=1}^{C_\ccalS}$ and assigned to the destination component $\ccalD_{A_i} \in \{\ccalD_{k}\}_{k=1}^{C_\ccalD}$ for $i\!=\!1,...,n$. 
		Then, if $\ccalE$ satisfies  	
		\begin{align}\label{eq:prioritizedOfflineCompletenessCondition}
			|\ccalE \!\setminus\! (\Delta \!\cup\! \ccalS \!\cup\! \ccalD)| \!\ge\! 2 d_{A_1} \hat{r} \!+\! \sum_{i=2}^n \!2 (d_{A_i,\ccalS}\!+\!d_{A_i,\ccalD})\hat{r}
		\end{align}		
		where $d_{A_1}$ is the distance between $\ccalS_{A_1}$ and $\ccalD_{A_1}$, i.e., $d_{A_1} = d_{\max} (\ccalS_{A_1}, \ccalD_{A_1})$, $d_{A_i,\ccalS}$ the minimal distance between $\ccalS_{A_i}$ and $\{\ccalS_{A_j}\}_{j=1}^{i-1}$, i.e., $d_{A_i,\ccalS} = \min \{ d_{\max} (\ccalS_{A_i}, \ccalS_{A_1}), \ldots, d_{\max} (\ccalS_{A_i}, \ccalS_{A_{i-1}})\}$, $d_{A_i,\ccalD}$ the minimal distance between $\ccalD_{A_i}$ and $\{\ccalD_{A_j}\}_{j=1}^{i-1}$, i.e., $d_{A_i,\ccalD} = \min \{ d_{\max} (\ccalD_{A_i}, \ccalD_{A_1}), \ldots, d_{\max} (\ccalD_{A_i}, \ccalD_{A_{i-1}})\}$, the offline prioritized environment optimization guarantees that the navigation tasks of all agents can be carried out successfully without collision.} \emph{Moreover, the traveled distance of agent $A_i$ outside $\ccalS$ and $\ccalD$ is bounded by
		\begin{align}\label{eq:performanceOfflinePriority}
			|\bbp_i| \le d_{A_1} + \sum_{j=2}^i (d_{A_i,\ccalS}+d_{A_i,\ccalD}) = C_i
		\end{align}
		for $i=1,\ldots,n$, which increases with the decreasing of priority, i.e., the distance bounds $C_1 \le \cdots \le C_n$ with the priorities $p_1 \ge \cdots \ge p_n$.
	}
\end{theorem}

\begin{proof}
	See Appendix \ref{proof:Theorem3}.
\end{proof}

Theorem \ref{thm:prioritizedOfflineCompleteness} states that the offline prioritized environment optimization guarantees the success of all navigation tasks and requires less obstacle-free area  
compared to its unprioritized counterpart, i.e., the lower bound in \eqref{eq:prioritizedOfflineCompletenessCondition} is smaller than that in \eqref{eq:offlineCompletenessCondition} [Asm. \ref{as:distance}]. However, this improvement is obtained by sacrificing the navigation performance of the agents with lower priorities. That is, these agents are no longer moving along the shortest path and the traveled distance increases with the decreasing of agent priority [cf. \eqref{eq:performanceOfflinePriority}]. It corresponds to Problem \ref{def:prioritizedEO} of prioritized environment optimization, where the metric $M(\cdot)$ is the traveled distance. The latter shows a trade-off between the navigation completeness of all agents and the individual performance of lower-priority agents. 

Next, we consider scenarios with further reduced resources, i.e., we show the partial completeness when the obstacle-free area is smaller than that required in \eqref{eq:prioritizedOfflineCompletenessCondition}. 

\begin{corollary}\label{coro:incompletePrioritizedOfflineEO}
	\emph{Consider the same setting as in Theorem \ref{thm:prioritizedOfflineCompleteness}. If the environment $\ccalE$ satisfies 	
		\begin{align}\label{eq:prioritizedOfflineIncompletenessCondition}
			&2 d_{A_1} \hat{r} + \sum_{i=2}^{b} 2 (d_{A_i,\ccalS}+d_{A_i,\ccalD})\hat{r}\\
			&\le |\ccalE \setminus (\Delta \cup \ccalS \cup \ccalD)| < 2 d_{A_1} \hat{r} + \sum_{i=2}^{b+1} 2 (d_{A_i,\ccalS}+d_{A_i,\ccalD})\hat{r} \nonumber
		\end{align}	
		where $b$ is an integer with $1 \le b < n$, the offline prioritized environment optimization guarantees that the navigation tasks of the agents with highest $b$ priorities $\{A_i\}_{i=1}^b$ can be carried out successfully without collision.} 
	
	\emph{For the rest of the agents $\{A_j\}_{j=b+1}^n$, if the starting position and the destination of agent $A_j$ is within the same connected components in $\ccalS$ and $\ccalD$ as one of the agents $\{A_i\}_{i=1}^b$, i.e., $\bbs_j \in \ccalS_{A_i}$ and $\bbd_j \in \ccalD_{A_i}$ for any $i \in \{1,\ldots,b\}$, 
		the navigation task of agent $A_j$ can be carried out successfully without collision as well.}
\end{corollary}

\begin{proof}
	See Appendix \ref{proof:Corollary1}.	
\end{proof}

Corollary \ref{coro:incompletePrioritizedOfflineEO} demonstrates that the success of all navigation tasks may not be guaranteed if the obstacle-free area is smaller than the lower bound in \eqref{eq:prioritizedOfflineCompletenessCondition}. In this case, the prioritized environment optimization emphasizes the navigation tasks of the agents with higher priorities $\{A_i\}_{i=1}^b$, but overlooks the ones with lower priorities $\{A_j\}_{j=b+1}^n$. Moreover, the navigation tasks of $\{A_j\}_{j=b+1}^n$ can be guaranteed only if their starting and goal positions are within the same connected components as one of the higher-priority agents $\{A_i\}_{i=1}^b$. This implies that the lower-priority agents can benefit from the higher-priority agents if they have a ``similar'' navigation task. 

\subsection{Online Prioritized Environment Optimization} 

Online environment optimization changes the obstacle region $\Delta$ during 
navigation and guarantees the success of all navigation tasks if $\dot{\Delta} \ge 2 n \hat{r} \| \hat{\bbv} \|_2$ [Thm. \ref{thm:onlineCompleteness}]. We similarly consider the reduced-resource scenario, where the capacity of the obstacle region is smaller than $2 n \hat{r} \| \hat{\bbv} \|_2$, i.e., $\dot{\Delta} < 2 n \hat{r} \| \hat{\bbv} \|_2$. In this circumstance, 
online prioritized environment optimization can guarantee the success of all navigation tasks, with a performance loss of lower-priority agents. 
\begin{theorem}\label{thm:prioritizedOnlineCompleteness}
	\emph{Consider the multi-agent system of $n$ agents $\{A_i\}_{i=1}^n$ in an environment $\ccalE$ with the same settings as Theorem \ref{thm:onlineCompleteness}. Let $\bbrho$ be the agent priorities ordered by the index 
		$\rho_1 \ge \rho_2 \ge \cdots \ge \rho_n$ 
		and the reduced capacity of environment optimization be 
		\begin{align}\label{eq:prioritizedOnlineCompletenessCondition}
			2 b \hat{r} \| \hat{\bbv} \|_2 \le \dot{\Delta} < 2 (b+1) \hat{r} \| \hat{\bbv} \|_2
		\end{align}	
		where $b$ is an integer with $1 \le b < n$. Then, online prioritized environment optimization guarantees that the navigation tasks of all agents can be carried out successfully without collision. Moreover, the navigation time of agent $A_i$ can be bounded proportionally to the inverse of its priority
		\begin{align}\label{eq:prioritizedNavigationTime}
			T_i \le \frac{C_T}{\rho_i},~\for~i=1,\ldots,n
		\end{align} 
		where $C_T$ is a time constant determined by $b$.}
\end{theorem}

\begin{proof}
	See Appendix \ref{proof:Theorem4}.
\end{proof}

Theorem \ref{thm:prioritizedOnlineCompleteness} states that the online prioritized environment optimization guarantees the success of all navigation tasks and requires a smaller obstacle changing rate compared to its unprioritized counterpart, i.e., the lower bound in \eqref{eq:prioritizedOnlineCompletenessCondition} is less than that in \eqref{eq:onlineCompletenessCondition}. This improvement comes with the performance loss of lower-priority agents, i.e., the agents with higher priorities reach the destinations faster than the ones with lower priorities [cf. \eqref{eq:prioritizedNavigationTime}]. We attribute this behavior to the fact that the online prioritized environment optimization changes the obstacle region concurrently to agent movement and is capable of continuously tuning the changing strategy based on the agent priorities during navigation. This corresponds to Problem \ref{def:prioritizedEO} of prioritized environment optimization, where the metric $M(\cdot)$ is the navigation time.

We follow to consider scenarios where all agents are required to reach their destinations \emph{within a required time}, and analyze the partial completeness of the online prioritized environment optimization in this setting.

\begin{corollary}\label{coro:IncompletePrioritizedOnlineEO}
	Consider the same setting as in Theorem \ref{thm:prioritizedOnlineCompleteness}. If all agents are required to reach their destinations within time $T_{\max}$, i.e., $T_i \le T_{\max}$ for $i=1,\ldots,n$, the online prioritized environment optimization guarantees the success of $n_p \le n$ agents $\{A_i\}_{i=1}^{n_p}$ with highest priorities $\{\rho_i\}_{i=1}^{n_p}$, where $n_p$ is determined by the changing rate of the obstacle region, i.e., $b$ in \eqref{eq:prioritizedOnlineCompletenessCondition}, and the required time $T_{\max}$ [cf. \eqref{eq:proofcoro27}]. 
\end{corollary}

\begin{proof}
	See Appendix \ref{proof:Corollary2}.
\end{proof}

Corollary \ref{coro:IncompletePrioritizedOnlineEO} demonstrates that the success of all navigation tasks may not be guaranteed if agents are required to reach destinations within a finite time $T_{\max}$. Online prioritized environment optimization allows a set of agents with higher priorities to reach destinations, while the rest of agents may fail navigation tasks. The number of successful agents depends on the changing capacity of the obstacle region and the required time $T_{\max}$, i.e., the higher the changing rate and the required time, the larger the number of successful agents. Corollary \ref{coro:IncompletePrioritizedOnlineEO} recovers Theorem \ref{thm:prioritizedOnlineCompleteness} when the required time is larger than the maximal time in \eqref{eq:prioritizedNavigationTime}, i.e., $T_{\max} \ge C_T/\rho_n$. 

Theorems \ref{thm:prioritizedOfflineCompleteness}-\ref{thm:prioritizedOnlineCompleteness} show the role played by agent priorities in the offline and online environment optimization with less resources, i.e., the obstacle-free area and the obstacle changing rate, than that required by Theorems \ref{thm:offlineCompleteness}-\ref{thm:onlineCompleteness}. The prioritized environment optimization guides the resource allocation by emphasizing higher-priority agents. By doing so, it guarantees the success of all navigation tasks with reduced resources while sacrificing the navigation performance of lower-priority agents. Moreover, Corollaries \ref{coro:incompletePrioritizedOfflineEO}-\ref{coro:IncompletePrioritizedOnlineEO} demonstrate the partial completeness of multi-agent navigation for scenarios with further reduced resources and added requirements (e.g., maximal time allowed). The prioritized environment optimization guarantees the success of higher-priority agents, but lower-priority agents may fail in these circumstances. 

\begin{remark}
	Theoretical analysis of offline and online environment optimization assumes that the obstacle region $\Delta$ can be controlled and re-shaped continuously as long as its area remains the same. The goal is to demonstrate the effectiveness of environment optimization to improve the performance of multi-agent navigation. However, it is worth mentioning that there may exist practical restrictions either on the obstacle region or on the way it can be controlled, in real-world applications. We account for the latter by imposing constraints when mathematically formulating the environment optimization problem in the next section. 
\end{remark}

\section{Methodology}\label{sec:constrainedEO}

In this section, we formulate the prioritized environment optimization problem mathematically as a constrained stochastic optimization problem, where the imposed constraints correspond to resource limitations and physical restrictions on the environment optimization in real-world applications, and solve the latter by leveraging reinforcement learning with the primal-dual mechanism. 

Specifically, consider the obstacle region $\Delta$  comprising $m$ obstacles $\ccalO = \{O_1,\ldots,O_m\}$, which can be of any shape and located at positions $\bbO = [\bbo_1,\ldots,\bbo_m]$. The obstacles can change positions to improve the navigation performance of the agents. Denote by $\bbX_o = [\bbx_{o1},\ldots,\bbx_{om}]$ the obstacle states, $\bbU_o = [\bbu_{o1},\ldots,\bbu_{om}]$ the obstacle actions, $\bbX_a = [\bbx_{a1},\ldots,\bbx_{an}]$ the agent states and $\bbU_a = [\bbu_{a1},\ldots,\bbu_{am}]$ the agent actions. For example, the states can be positions or velocities and the actions can be accelerations. Let $\pi_o(\bbU_o | \bbX_o, \bbX_a)$ be an optimization policy that controls the obstacles, a distribution over $\bbU_o$ conditioned on $\bbX_o$ and $\bbX_a$. The objective function $f(\bbS, \bbD, \pi_a, \bbrho, \bbO, \pi_o)$ measures the performance of the multi-agent navigation task $(\bbS, \bbD)$, given the trajectory planner $\pi_a$, the agent priorities $\bbrho$, the obstacle positions $\bbO$ and the optimization policy $\pi_o$, while the cost functions $\{g_j(\bbS, \bbD, \pi_a, \bbrho, \bbO, \pi_o)\}_{j=1}^m$ indicate the penalties of position changes w.r.t. the obstacles. Moreover, a set of $Q$ constraints are imposed on the obstacles corresponding to resource limitations and physical restrictions, such as limitations of moving velocity and distance, which are represented by the constraint functions $\{h_q(\bbX_o, \bbU_o)\}_{q=1}^Q$. Since the multi-agent system is with random initialization, the objective, cost, constraint functions are random functions and an \textit{expected} performance would be a more meaningful metric for performance evaluation. 

The goal is, thus, to find an optimal obstacle policy $\pi_o$ that maximizes the expected performance $\mathbb{E}[f(\bbS, \bbD, \pi_a, \bbrho, \bbO, \pi_o)]$ regularized by the obstacle costs $\{\mathbb{E}[g_j(\bbS, \bbD, \pi_a, \bbrho, \bbO, \pi_o)]\}_{j=1}^m$, while satisfying the imposed constraints. By introducing $\ccalU_o$ as the action space of the obstacles, we formulate the prioritized environment optimization problem as a constrained stochastic optimization problem
\begin{align}\label{eq:environmentProblemMath}
	&\argmax_{\pi_o}~\!\! \mathbb{E}[f(\bbS,\! \bbD,\! \pi_a,\! \bbrho,\! \bbO,\! \pi_o\!)\!] \!-\!\!\! \sum_{j=1}^m\! \beta_j \mathbb{E}[g_j(\bbS,\! \bbD,\! \pi_a,\! \bbrho,\! \bbO,\! \pi_o\!)\!] \nonumber\\
	&~~~~\text{s.t.}~~~~ \mathbb{E}[h_q(\bbX_o, \bbU_o)] \le 0,~\forall~q\!=\!1,\ldots,Q,\nonumber\\
	&~~~~~~~~~~~~ \pi_{o}(\bbU_{o}|\bbX_o,\bbX_a) \in \ccalU_{o}, 
\end{align}
where $\{\beta_j\}_{j=1}^m$ are regularization parameters. The objective function $f(\bbS, \bbD, \pi_a, \bbrho, \bbO, \pi_o)$, the cost functions $\{g_j(\bbS, \bbD, \pi_a, \bbrho, \bbO, \pi_o)\}_{j=1}^m$, the constraint functions $\{h_q(\bbX_o, \bbU_o)\}_{q=1}^Q$ and the action space $\ccalU_{o}$ are not necessarily convex/non-convex depending on specific application scenarios. The problem is challenging due to four main reasons: 
\begin{enumerate}[(i)]
	
	\item The closed-form expression of the objective function $f(\bbS, \bbD, \pi_a, \bbrho, \bbO, \pi_o)$ is difficult to derive because the explicit relationship between agents, obstacles and navigation performance is difficult to model, precluding the application of conventional model-based methods and leading to poor performance of heuristic methods.
	
	\item The imposed constraints $\{h_q(\bbX_o, \bbU_o)\}_{q=1}^Q$ restrict the space of feasible solutions and are difficult to address, leading to the failure of conventional unconstrained optimization algorithms.
	
	\item The policy $\pi_o(\bbU_o | \bbX_o, \bbX_a)$ is an arbitrary mapping from the state space to the action space, which can take any function form and is infinitely dimensional. 
	
	\item The obstacle actions can be discrete or continuous and the action space $\ccalU_o$ can be non-convex, resulting in further constraints on the feasible solution.
	
\end{enumerate}
Due to the aforementioned challenges, we propose to solve problem \eqref{eq:environmentProblemMath} by leveraging reinforcement learning (RL) and the primal-dual mechanism. The former parameterizes the optimization policy with information processing architectures and allows us to train the architecture parameters in a model-free manner. The latter penalizes the constraints with dual variables, and updates the primal and dual variables alternatively while searching for feasible solutions.

\subsection{Reinforcement Learning} 

We reformulate problem \eqref{eq:environmentProblemMath} in the RL domain and start by defining a Markov decision process. At each step $t$, the obstacles $\ccalO$ and the agents $\ccalA$ observe the states $\bbX_o^{(t)}$, $\bbX_a^{(t)}$ and take the actions $\bbU_o^{(t)}$, $\bbU_a^{(t)}$ with the obstacle policy $\pi_o$ and the trajectory planner $\pi_a$, respectively. The actions $\bbU_o^{(t)}$, $\bbU_a^{(t)}$ amend the states $\bbX_o^{(t)}$, $\bbX_a^{(t)}$ based on the transition probability function $P(\bbX_o^{(t+1)},\bbX_a^{(t+1)}|\bbX_o^{(t)},\bbX_a^{(t)},\bbU_o^{(t)},\bbU_a^{(t)})$, which is a distribution over the states conditioned on the previous states and the actions. Let $r_{ai}(\bbX_o^{(t)},\bbX_a^{(t)}, \bbU_o^{(t)},\bbU_a^{(t)})$ be the reward function of agent $A_i$, which represents its navigation performance at step $t$. The reward function of obstacle $O_j$ comprises two components: (i) the global system reward and (ii) the local obstacle reward, i.e., 
\begin{align}\label{eq:obstacleRewardFunction}
	&r_{oj} = \frac{1}{n}\sum_{i=1}^n \rho_i r_{ai} + \beta_j r_{j, local},~\forall~j=1,...,m
\end{align}
where $\rho_i$ is the priority of agent $A_i$, 
$\beta_j$ is the regularization parameter of obstacle $O_j$, and $r_{ai}$, $r_{j, local}$ are concise notations of $r_{ai}(\bbX_o^{(t)},\bbX_a^{(t)}, \bbU_o^{(t)},\bbU_a^{(t)})$, $r_{j, local} (\bbX_o^{(t)},\bbX_a^{(t)}, \bbU_o^{(t)},\bbU_a^{(t)})$. The first term is the team reward that represents the multi-agent system performance, which is shared over all obstacles. The second term is the individual reward that corresponds to the position change penalty of obstacle $O_j$, which may be different among obstacles, e.g., the collision penalty. This reward definition is novel that differs from common RL scenarios, which is a combination of a global reward with a local reward. The former is the main goal of all obstacles, while the latter is the individual cost of an obstacle. The priorities $\bbrho$ weigh the agents' performance to put more emphasis on the agents with higher priorities, while the regularization parameters $\{\beta_j\}_{j=1}^m \in [0, 1]^m$ weigh the environment optimization penalty on the navigation performance. The total reward of the obstacles is defined as 
\begin{align}\label{eq:obstacleReward}
	r_{o}= \sum_{j=1}^m r_{oj} = \frac{m}{n}\sum_{i=1}^n \rho_i r_{ai} + \sum_{j=1}^m\beta_j r_{j, local}.
\end{align}

Let $\gamma \in [0,1]$ be the discount factor accounting for the future rewards and the expected discounted reward can be represented as
\begin{align}\label{eq:expectedCost}
	&R(\bbS, \bbD, \pi_a, \bbrho, \bbO, \pi_o) = \mathbb{E} \Big[\sum_{t = 0}^\infty \gamma^{t} r_{o}^{(t)}\Big] \\
	& = \mathbb{E}\Big[\sum_{t=0}^\infty \gamma^t \sum_{j=1}^m\!\sum_{i=1}^n\! \frac{\rho_i r_{ai}^{(t)}}{n}\Big] \!+\! \sum_{j=1}^m \beta_j \mathbb{E}\Big[\sum_{t=0}^\infty \gamma^t r_{j, local}^{(t)}\Big] \nonumber
\end{align}
where $\bbO$, $\bbS$ and $\bbD$ are the initial positions and destinations that determine the initial states $\bbX_o^{(0)}$ and $\bbX_a^{(0)}$, and $\mathbb{E}[\cdot]$ is w.r.t. the action policy and the state transition probability. The obstacles are imposed with $Q$ constraints at each step $t$, which are functions of the obstacle states $\bbX_o^{(t)}$ and actions $\bbU_o^{(t)}$ as $h_q(\bbX_o^{(t)}, \bbU_o^{(t)}) \le 0$ for $q=1,\ldots,Q$. By parameterizing the obstacle policy $\pi_o$ with an information processing architecture $\bbPhi(\bbX_o,\bbX_a,\bbtheta)$ of parameters $\bbtheta$, we formulate the constrained reinforcement learning problem as
\begin{align}\label{eq:learningProblem}
	&~~\argmax_{\bbtheta}~ R(\bbS, \bbD, \pi_a, \bbrho, \bbO, \bbtheta)\\
	&~~~\text{s.t.}~~h_q(\bbX_o^{(t)}, \bbU_o^{(t)}) \le 0,~\for~q\!=\!1,...,Q,~t=1,\ldots,\infty,\nonumber\\
	&~~~~~~~~\bbPhi(\bbX_o^{(t)}, \bbX_a^{(t)}, \bbtheta) \in \ccalU_{o},~\for~t=1,\ldots,\infty.\nonumber
\end{align}
By comparing problem \eqref{eq:environmentProblemMath} with problem \eqref{eq:learningProblem}, we see equivalent representations of the objective, cost and constraint functions in the RL domain. 
The goal is to find the optimal parameters $\bbtheta^*$ that maximize the reward while satisfying the constraints at each step $t$. 

\subsection{Primal-Dual Policy Gradient}

Since there is no straightforward way to optimize $\bbtheta$ w.r.t. per-step hard constraints in problem \eqref{eq:learningProblem}, we define the reward of the constraint $h_q(\bbX_o^{(t)}, \bbU_o^{(t)}) \le 0$ with an indicator function as
\begin{align}\label{eq:indicatorFunction}
	r_q^{(t)} = \mathbbm{1}[h_q(\bbX_o^{(t)}, \bbU_o^{(t)}) \le 0],~\for~q=1,\ldots,Q
\end{align} 
where $\mathbbm{1}[h_q(\bbX_o^{(t)}, \bbU_o^{(t)}) \le 0]$ is $1$ if $h_q(\bbX_o^{(t)}, \bbU_o^{(t)}) \le 0$ and $0$ otherwise. It quantifies the constraint by rewarding success and penalizing failure at each step $t$. The cumulative reward of the constraint $h_q(\bbX_o^{(t)}, \bbU_o^{(t)}) \le 0$ is 
\begin{align}\label{eq:expectedConstraintReward}
	R_q(\bbO, \bbtheta) = \mathbb{E}\Big[\sum_{t=1}^\infty \gamma^t r_q^{(t)}\Big],~\for~q=1,\ldots,Q
\end{align} 
where $\gamma$ is the discount factor. The preceding expression allows us to measure how well the constraints are satisfied in expectation and takes the same form as the expected discounted reward of the obstacles [cf. \eqref{eq:expectedCost}]. We can then transform the constraints as
\begin{align}\label{eq:constraintConstants}
	R_q(\bbO, \bbtheta) \ge \ccalC_q,~\for~q=1,\ldots,Q
\end{align}
where $\{\ccalC_q\}_{q=1}^Q$ are constants that lower-bound the constraint rewards and control the constraint guarantees -- see details in Sec. \ref{subsec:constraintGuarantee}. This yields an alternative of problem \eqref{eq:learningProblem} as
\begin{align}\label{eq:alternativeLearningProblem}
	&~~~\mathbb{P}:=\argmax_{\bbtheta}~ R(\bbS, \bbD, \pi_a, \bbrho, \bbO, \bbtheta)\\
	&~~~\text{s.t.}~~R_q(\bbO, \bbtheta) \ge \ccalC_q,~\for~q=1,\ldots,Q,\nonumber\\
	&~~~~~~~~\bbPhi(\bbX_o^{(t)}, \bbX_a^{(t)}, \bbtheta) \in \ccalU_{o},~\for~t=1,\ldots,\infty.\nonumber
\end{align}

By introducing the dual variables $\bblambda = [\lambda_1,\ldots,\lambda_Q]^\top \in \mathbb{R}_+^Q$ that correspond to $Q$ constraints, we formulate the Lagrangian of problem \eqref{eq:alternativeLearningProblem} as
\begin{align}\label{eq:Lagrangian}
	&\ccalL(\bbtheta, \bblambda) \!=\! R(\bbS,\! \bbD,\! \pi_a,\! \bbrho,\! \bbO,\! \bbtheta) \!+\!\! \sum_{q=1}^Q \!\lambda_q \big(R_q(\bbO, \bbtheta) \!-\! \ccalC_q\big) 
\end{align}
which penalizes the objective with the constraint violation weighted by the dual variables $\bblambda$. The dual function is defined as the maximum of the Lagrangian $\ccalD(\bblambda):=\max_{\bbtheta} \ccalL(\bbtheta, \bblambda)$. Since $\ccalD(\bblambda) \ge \mathbb{P}$ for any dual variables $\bblambda$, we define the dual problem as
\begin{align}\label{eq:dualProblem}
	\mathbb{D}:= \min_{\bblambda \ge 0} \ccalD(\bblambda) = \min_{\bblambda \ge 0} \max_{\bbtheta} \ccalL(\bbtheta, \bblambda)
\end{align}
which computes the optimal dual variables that minimizes the dual function $\ccalD(\bblambda)$. The dual solution $\mathbb{D}$ in \eqref{eq:dualProblem} can be considered as the best approximation of the primal solution $\mathbb{P}$ in \eqref{eq:alternativeLearningProblem}. This translates a constrained maximization problem to an unconstrained min-max problem, which searches for a saddle point solution $(\bbtheta^*, \bblambda^*)$ that is maximal w.r.t. the primal variables $\bbtheta$ and minimal w.r.t. the dual variables $\bblambda$. 

We propose to solve the dual problem \eqref{eq:dualProblem} with the primal-dual policy gradient method, which updates $\bbtheta$ with policy gradient ascent and $\bblambda$ with gradient descent in an alternative manner. Specifically, it trains the model iteratively and each iteration consists of primal and dual steps. 

\smallskip
\noindent \textbf{Primal step.} At iteration $k$, let $\bbtheta^{(k)}$ be the primal variables, $\bblambda^{(k)}$ the dual variables, $\bbX_o^{(k)}$ the obstacle states, and $\bbX_a^{(k)}$ the agent states. Given these system states, the obstacle policy $\bbPhi(\bbX_o^{(k)},\bbX_a^{(k)},\bbtheta^{(k)})$ generates the obstacle actions $\bbU_o^{(k)}$ and the given trajectory planner $\pi_a$ generates the agent actions $\bbU_a^{(k)}$. These actions $\bbU_o^{(k)}$, $\bbU_a^{(k)}$ change the states from $\bbX_o^{(k)}$, $\bbX_a^{(k)}$ to $\bbX_{o,1}^{(k)}$, $\bbX_{a,1}^{(k)}$ based on the transition probability function $P(\bbX_{o,1}^{(k)},\bbX_{a,1}^{(k)}|\bbX_o^{(k)},\bbX_a^{(k)},\bbU_o^{(k)},\bbU_a^{(k)})$, which feeds back the corresponding obstacle reward $r_o^{(k)}$ [cf. \eqref{eq:obstacleReward}] and the constraint rewards $\{r_q^{(k)}\}_{q=1}^Q$ [cf. \eqref{eq:indicatorFunction}]. 

We follow the actor-critic method to define a value function $V(\bbX_o,\bbX_a, \bbomega)$ of parameters $\bbomega$ that estimates the Lagrangian \eqref{eq:Lagrangian} initialized at the states $\bbX_o,\bbX_a$. Let $\bbomega^{(k)}$ be the parameters of the value function at iteration $k$, and $V(\bbX_o^{(k)},\bbX_a^{(k)}, \bbomega^{(k)})$, $V(\bbX_{o,1}^{(k)},\bbX_{a,1}^{(k)}, \bbomega^{(k)})$ be the estimated values at $\bbX_o^{(k)},\bbX_a^{(k)}$ and $\bbX_{o,1}^{(k)},\bbX_{a,1}^{(k)}$. This allows us to compute the estimation error as
\begin{align}\label{eq:valueFunctionError}
	\delta^{(k)}_1 \!&=\! r_o^{(k)} \!\!+\! \sum_{q=1}^Q \lambda_q^{(k)}\big(r_q^{(k)} - (1-\gamma)\ccalC_q\big) \\ &+ \gamma V(\bbX_{o,1}^{(k)}\!,\bbX_{a,1}^{(k)}\!, \bbomega^{(k)}) \!-\! V\big(\bbX_o^{(k)}\!,\bbX_a^{(k)}\!, \bbomega^{(k)}\big) \nonumber
\end{align}
which is used to update the parameters $\bbomega^{(k)}$ of the value function. Then, we can update the primal variables $\bbtheta^{(k)}$ with policy gradient ascent as
\begin{align}
	\label{eq:policyGradientObstacle} \bbtheta^{(k)}_{1} = \bbtheta^{(k)} + \alpha \nabla_{\bbtheta} \log \pi_o(\bbU_o^{(k)} | \bbX_a^{(k)}, \bbX_o^{(k)}) \delta^{(k)}_1
\end{align}
where $\alpha$ is the step-size and $\pi_o(\bbU_o^{(k)} | \bbX_a^{(k)}, \bbX_o^{(k)})$ is the policy distribution of the obstacle action specified by the parameters $\bbtheta^{(k)}$. The primal update is model-free because \eqref{eq:policyGradientObstacle} requires only the error value $\delta^{(k)}_1$ and the gradient of the policy distribution, but not the objective, cost and constraint function models. By performing the above procedure recursively $\psi \in \mathbb{Z}_+$ times, we obtain a sequence of primal variables $\{\bbtheta^{(k)}_{1}, \bbtheta^{(k)}_{2},\ldots,\bbtheta^{(k)}_{\psi}\}$. Define $\bbtheta^{(k+1)} := \bbtheta^{(k)}_{\psi}$ as the new primal variables and step into the dual step.

\begin{figure*}
	\centering
	\subfloat[]{\includegraphics[width=0.25\linewidth]{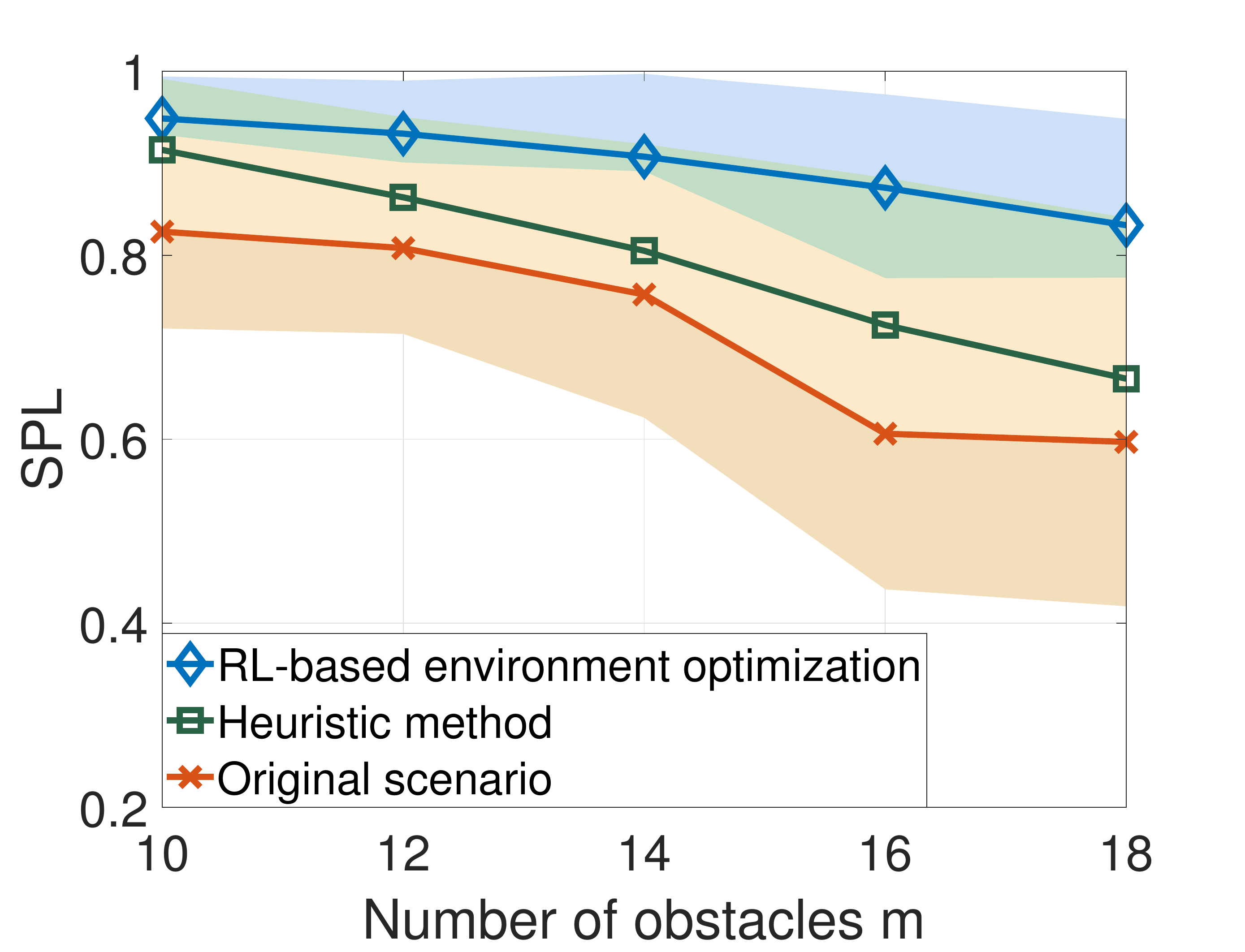}}
	\subfloat[]{\includegraphics[width=0.25\linewidth]{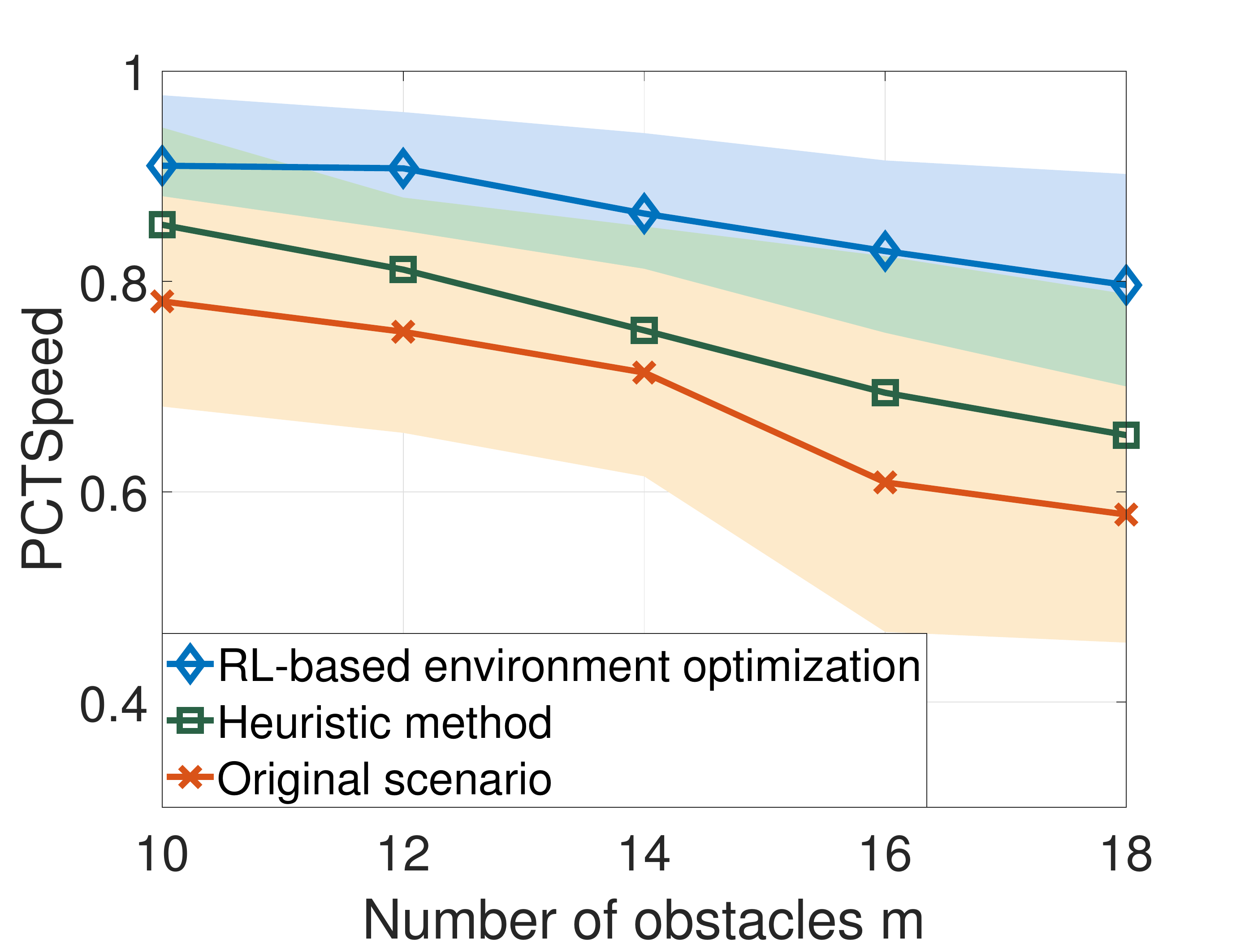}}
	\subfloat[]{\includegraphics[width=0.25\linewidth, height = 0.19\linewidth, trim = {1cm 0.7cm 2cm 0cm}, clip]{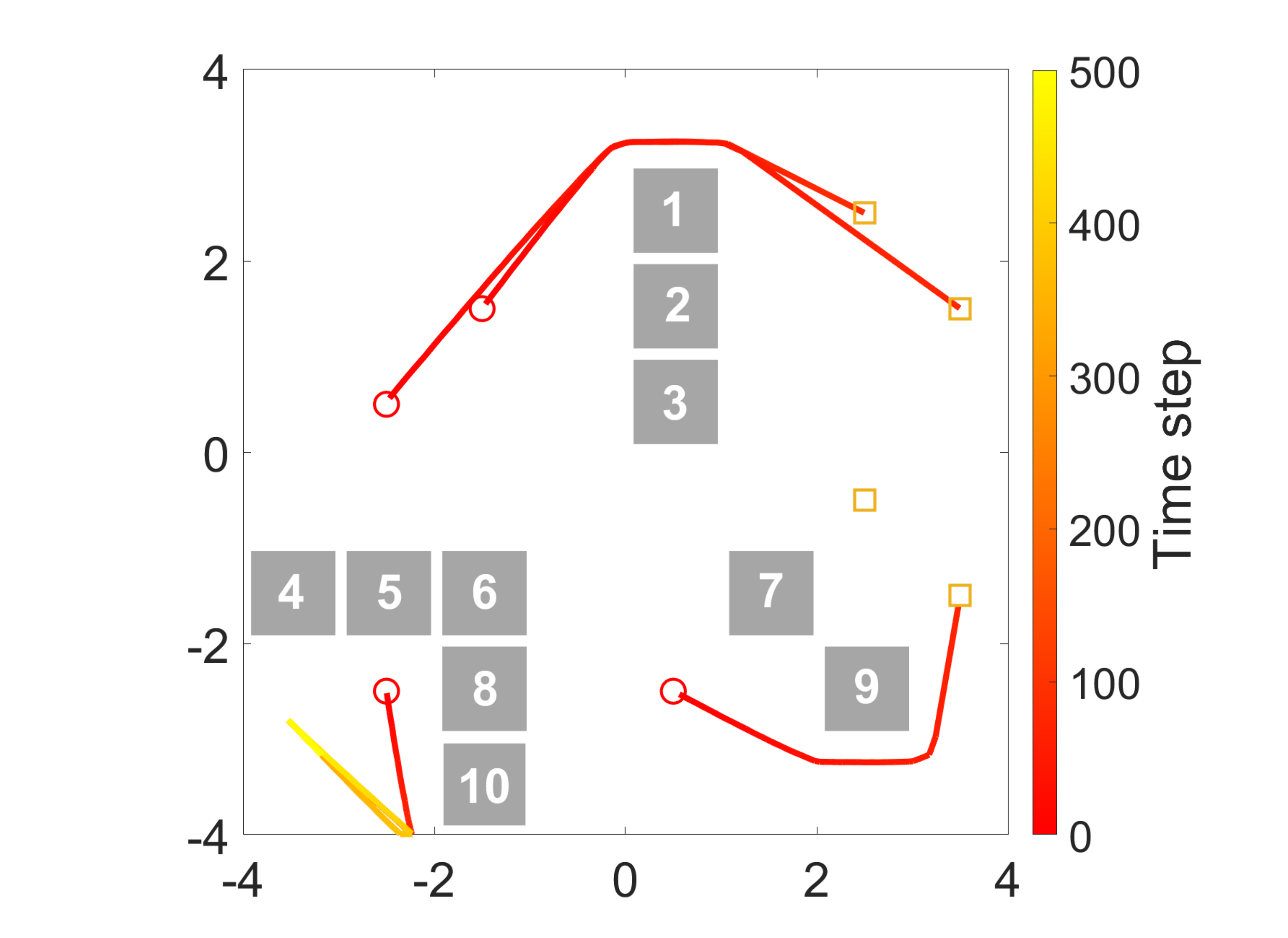}}
	\subfloat[]{\includegraphics[width=0.25\linewidth, height = 0.19\linewidth, trim = {1cm 0.7cm 2cm 0cm}, clip]{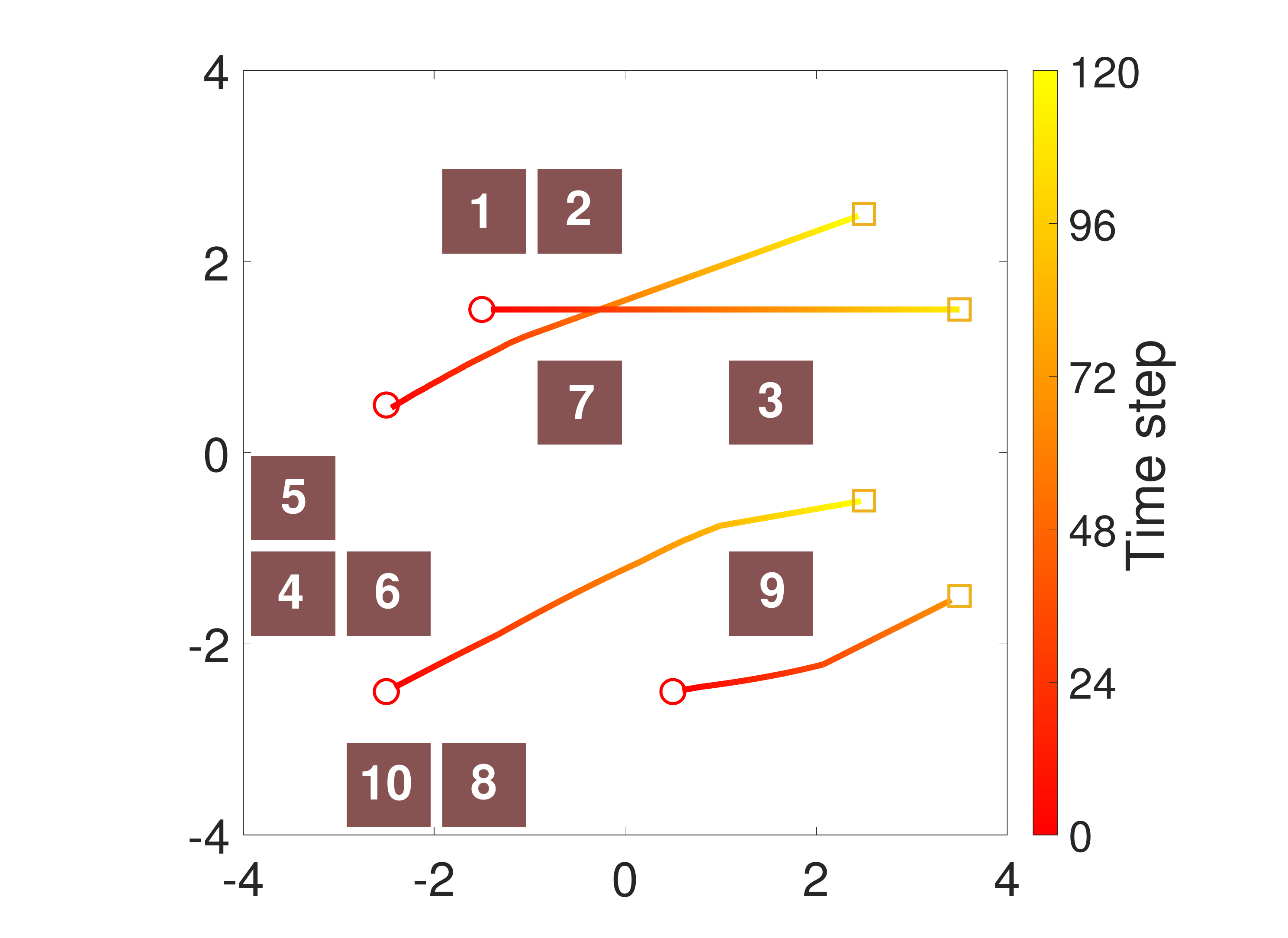}}
	\caption{(a-b) Performance of offline environment optimization compared to the baselines. Results are averaged over $20$ trials and the shaded area shows the std. dev. The RL system is trained on $10$ obstacles and tested on $10$ to $18$ obstacles. (a) SPL ($1$ is best). (b) PCTSpeed. (c-d) Example of offline environment optimization. Red circles are initial positions, yellow squares are destinations, and obstacles are numbered for exposition. Color lines from red to yellow are agent trajectories and the color bar represents the time scale. Obstacle layout and agent trajectories (c) \textit{before} environment optimization and (d) \textit{after} environment optimization.}\label{fig:offline}\vspace{-4mm}
\end{figure*}

\smallskip
\noindent \textbf{Dual step.} With the updated primal variables $\bbtheta^{(k+1)}$, we update the dual variables $\bblambda^{(k)}$ with gradient descent as
\begin{align}\label{eq:dualStep}
	&\lambda_q^{(k\!+\!1)} \!\!=\!\! \Big[\!\lambda_q^{(k)} \!\!-\!\! \beta \big(\!R_q(\bbO, \bbtheta^{(k\!+\!1)}) \!-\! \ccalC_q\big)\!\Big]_+\for~q\!=\!1,...,Q 
\end{align}
where $\beta$ is the step-size and $[\cdot]_+$ is the non-negative operator corresponding to the positivity of dual variables. Since the constraint function values $\{h_q(\bbX_o^{(k)}, \bbU_o^{(k)})\}_{q=1}^Q$ are given, the cumulative constraint rewards $\{R_q(\bbO, \bbtheta^{(k)})\}_{q=1}^Q$ can be estimated by sampling $\ccalT$ trajectories $\{\bbX_o^{(t),1}, \bbU_o^{(t), 1}\}_{t},\ldots,\{\bbX_o^{(t),\ccalT}, \bbU_o^{(t), \ccalT}\}_{t}$, computing the respective cumulative constraint rewards, and taking the average, i.e.,
\begin{align}\label{eq:constarintApprox}
	R_q(\bbO, \bbtheta^{(k\!+\!1)}) \!\approx\! \frac{1}{\ccalT}\sum_{\tau=1}^{\ccalT}\sum_{t=0}^{T_\tau} \!\gamma^t \mathbbm{1}[h_q(\bbX_o^{(t),\tau}\!, \bbU_o^{(t), \tau}) \!\le\! 0]
\end{align} 
where $T_\tau$ is the operation time of the $\tau$th trajectory. The dual update is model-free because \eqref{eq:dualStep}-\eqref{eq:constarintApprox} require only the constraint values $\{h_q(\bbX_o^{(t)}, \bbU_o^{(t)})\}_t$ instead of the constraint function models. Algorithm \ref{alg:primalDualPolicyGradient} summarizes the approach.

{\linespread{1}
	\begin{algorithm}[t] \begin{algorithmic}[1]
			\STATE \textbf{Input:} initial primal variables $\bbtheta^{(0)}$, initial dual variables $\bblambda^{(0)}$, initial value function parameters $\bbomega^{(0)}$, trajectory planner $\pi_a$ and transition probability function $P$
			\FOR {$k=1,\dots$}
			\STATE Compute value functions $V(\bbX_o^{(k)},\bbX_a^{(k)}, \bbomega^{(k)})$ and $V(\bbX_o^{(k+1)},\bbX_a^{(k+1)}, \bbomega^{(k)})$ to estimate the Lagrangian values of \eqref{eq:Lagrangian}
			\STATE Compute the estimation error as in \eqref{eq:valueFunctionError}
			\STATE Update the primal variables $\bbtheta^{(k)}$ with policy gradient ascent as in \eqref{eq:policyGradientObstacle} $\psi$ times and set $\bbtheta^{(k+1)} := \bbtheta^{(k)}_\psi$ as the updated primal variables 
			\STATE Compute the expected constraint rewards as in \eqref{eq:constarintApprox}
			\STATE Update the dual variables $\bblambda^{(k)}$ with gradient descent as in \eqref{eq:dualStep}
			\ENDFOR
		\end{algorithmic}
		\caption{Primal-Dual Policy Gradient Method}\label{alg:primalDualPolicyGradient}
\end{algorithm}}

\subsection{Constraint Guarantees}\label{subsec:constraintGuarantee}

While the cumulative constraint rewards in the alternative problem \eqref{eq:alternativeLearningProblem} is a relaxation of the per-step hard constraints in the original problem \eqref{eq:learningProblem}, we provide constraint guarantees for the alternative problem \eqref{eq:alternativeLearningProblem} in the sequel. Specifically, the expectation of the indicator function in \eqref{eq:expectedConstraintReward} is equivalent to the probability of satisfying the constraint, i.e.,
\begin{align}\label{eq:probability}
	\mathbb{E}\Big[\mathbbm{1}[h_q(\bbX_o^{(t)}, \bbU_o^{(t)}) \le 0]\Big] = \mathbb{P}\big[h_q(\bbX_o^{(t)}, \bbU_o^{(t)}) \le 0\big]
\end{align} 
where $\mathbb{P}[\cdot]$ is the probability measure. This allows to rewrite the cumulative constraint reward as
\begin{align}\label{eq:averageGuarantee}
	R_q(\bbO, \bbtheta) = \sum_{t=0}^\infty \gamma^t \mathbb{P}[h_q(\bbX_o^{(t)}, \bbU_o^{(t)}) \le 0] \ge \ccalC_q 
\end{align} 
for $q=1,\ldots,Q$ [cf. \eqref{eq:constraintConstants}], which is the discounted sum of the constraint satisfaction probability over all steps $t$. 

However, \eqref{eq:averageGuarantee} are not sufficient conditions to satisfy the constraints at each step $t$, as required by problem \eqref{eq:learningProblem}. We then establish the relation between \eqref{eq:averageGuarantee} and the per-step constraints in problem \eqref{eq:learningProblem}. Before proceeding, we define the $(1-\delta)$-constrained solution as follows. 
\begin{definition}[$(1-\delta)$-constrained solution]
	A solution of problem \eqref{eq:alternativeLearningProblem} is $(1-\delta)$-constrained w.r.t. the constraints $\{h_q\}_{q=1}^Q$ if for each step $t \ge 0$, it holds that
	\begin{align}
		\mathbb{P}\big[\!\cap_{0 \le \tau \le t}\! \{h_q(\bbX_o^{(\tau)}\!\!,\! \bbU_o^{(\tau)}\!) \!\le\! 0\}\big] \!\!\ge\!\! 1\!-\!\delta~\for~q\!=\!\!1,...,Q.
	\end{align}
\end{definition}

\smallskip
The $(1-\delta)$-constrained solution satisfies the constraints at each step $t$ with a probability $1 - \delta \in (0,1]$, which is a feasible solution of problem \eqref{eq:learningProblem} if $\delta = 0$. With these preliminaries, we analyze the constraint guarantees of the alternative problem \eqref{eq:alternativeLearningProblem} in the following theorem. 
\begin{theorem}\label{thm:safetyGuarantee}
	Consider problem \eqref{eq:alternativeLearningProblem} with the constraint constant taking the form of $\ccalC_q = (1-\delta + \eps)/(1-\gamma)$ for $q=1,\ldots,Q$ [cf. \eqref{eq:constraintConstants}], where $\delta$ and $\eps$ are constraint parameters with $0 \le \eps \le \delta$. For any $\delta$, there exists $\eps$ such that the solution of problem \eqref{eq:alternativeLearningProblem} is $(1 - \delta)$-constrained. 
\end{theorem}
\begin{proof}
	See Appendix \ref{proof:Theorem5}.
\end{proof}
Theorem \ref{thm:safetyGuarantee} states that the feasible solution of the alternative problem \eqref{eq:alternativeLearningProblem} is a $(1-\delta)$-constrained solution of problem \eqref{eq:learningProblem}. The result provides the per-step guarantees for the environment constraints in probability,  which demonstrates the applicability of the alternative problem \eqref{eq:alternativeLearningProblem}. In essence, we have not lost the constraint guarantees by relaxing the hard constraints to the cumulative constraint rewards. 

\subsection{Information Processing Architecture}

The proposed approach is a general framework that covers various environment optimization scenarios and different information processing architectures can be integrated to solve different variants. We illustrate this fact by analyzing two representative scenarios: 
(i) offline optimization in discrete environments and (ii) online optimization in continuous environments, for which we use convolutional neural networks (CNNs) and graph neural networks (GNNs), respectively. 

\smallskip
\noindent \textbf{CNNs for offline discrete settings.} In this setting, we first optimize the obstacles’ positions and then navigate the agents. Since computations are performed offline, we collect the states of all obstacles $\bbX_o$ and agents $\bbX_a$ apriori, which allows for centralized information processing solutions (e.g., CNNs). CNNs leverage convolutional filters to extract features from image signals and have found wide applications in computer vision \cite{browne2003convolutional, kumra2017robotic, gu2018recent}. In the discrete environment, the system states can be represented by matrices and the latter are equivalent to image signals. This motivates to consider CNNs for policy parameterization.

\smallskip
\noindent \textbf{GNNs for online continuous settings.} In this setting, obstacle positions change while agents move, i.e., the obstacles take actions instantaneously. In large-scale systems with real-time communication and computation constraints, obstacles may not be able to collect the states of all other obstacles/agents and centralized solutions may not be applicable. This requires a decentralized architecture that can be implemented with local neighborhood information. GNNs are inherently decentralizable and are, thus, a suitable candidate.

GNNs are layered architectures that leverage a message passing mechanism to extract features from graph signals \cite{scarselli2008graph, velivckovic2018graph, gao2021stochastic}. At each layer $\ell$, let $\bbX_{\ell-1}$ be the input signal. The output signal is generated with the message aggregation function $\ccalF_{\ell m}$ and the feature update function $\ccalF_{\ell u}$ as
\begin{align}
	[\bbX_{\ell}]_i \!=\! \ccalF_{\ell u}\Big( [\bbX_{\ell\!-\!1}]_i,\!\! \sum_{j \in \ccalN_i}\!\! \ccalF_{\ell m}\big( [\bbX_{\ell\!-\!1}]_i, [\bbX_{\ell\!-\!1}]_j, [\bbE]_{ij} \big)\! \Big)\nonumber
\end{align}
where $[\bbX_{\ell}]_i$ is the signal value at node $i$, $\ccalN_i$ are the neighboring nodes within the communication radius, $\bbE$ is the adjacency matrix, and $\ccalF_{\ell m}$, $\ccalF_{\ell u}$ have learnable parameters $\bbtheta_{\ell m}$, $\bbtheta_{\ell u}$. The output signal is computed with local neighborhood information only, and each node can make decisions based on its own output value; hence, allowing for decentralized implementation \cite{tolstaya2020learning, li2020graph, gao2022wide, gao2022decentralized}.

\begin{figure*}
	\centering
	\subfloat[]{\includegraphics[width=0.33\linewidth]{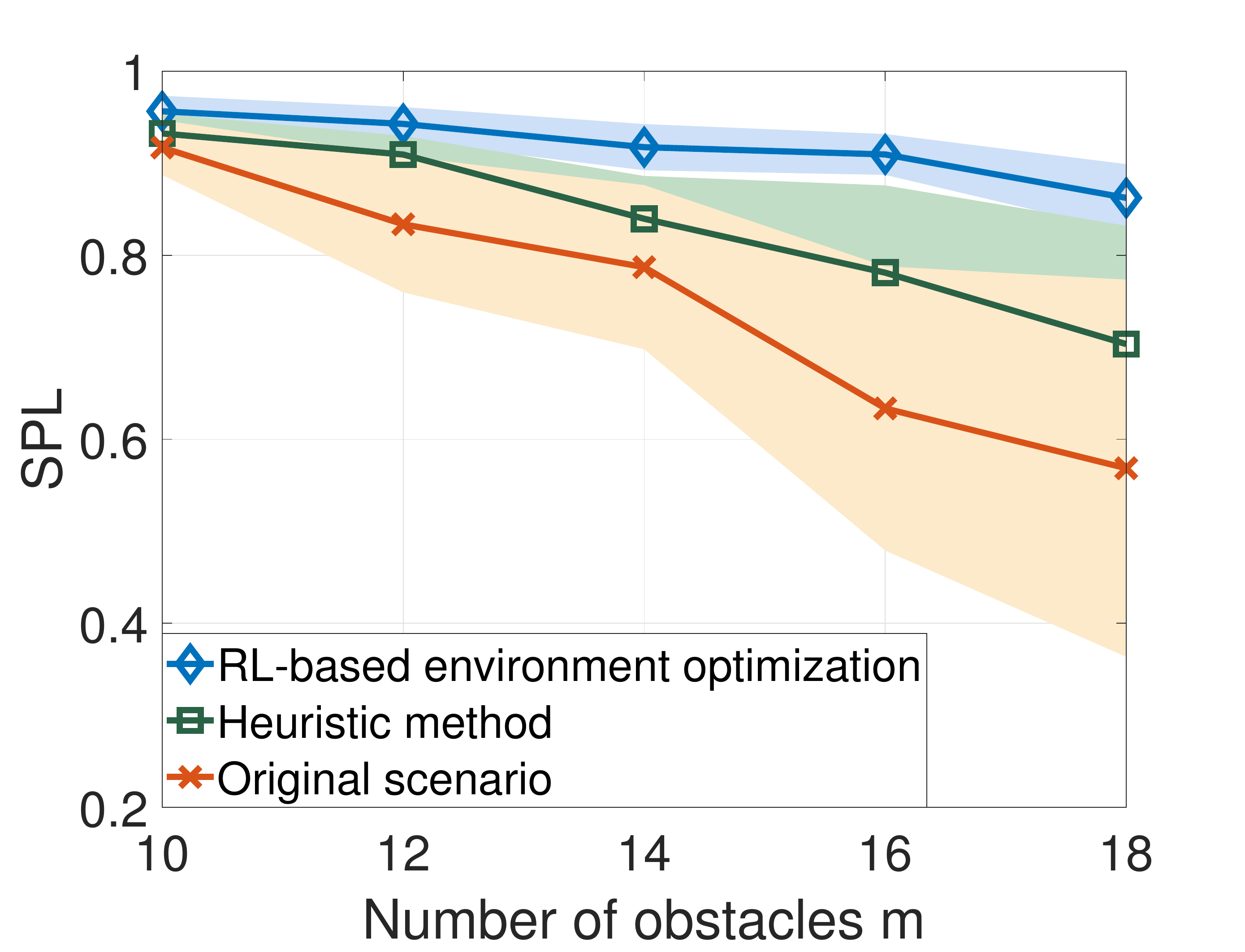}}
	\subfloat[]{\includegraphics[width=0.33\linewidth]{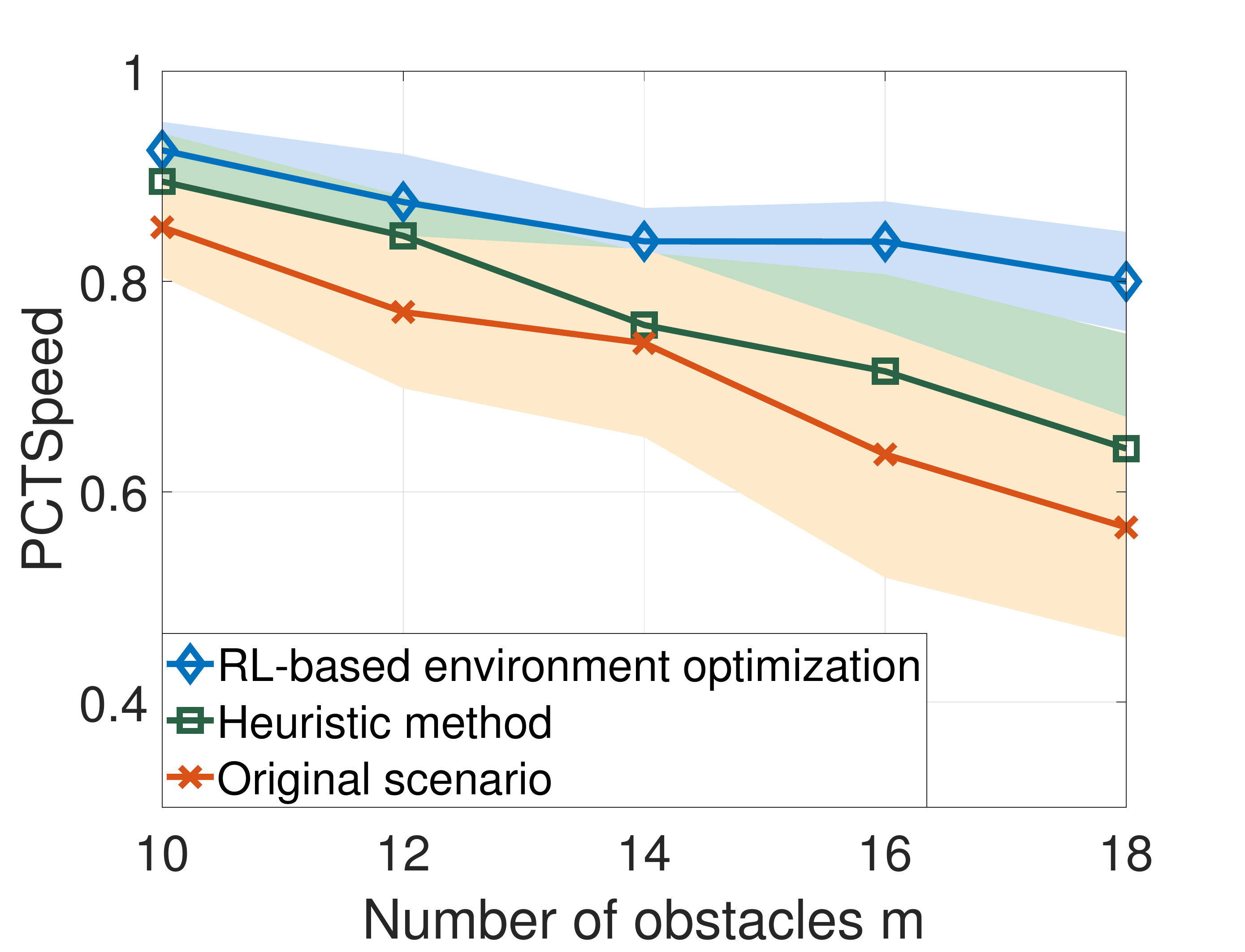}}
	\subfloat[]{\includegraphics[width=0.34\linewidth, height = 0.25\linewidth, trim = {0 0.7cm 2cm 0cm}, clip]{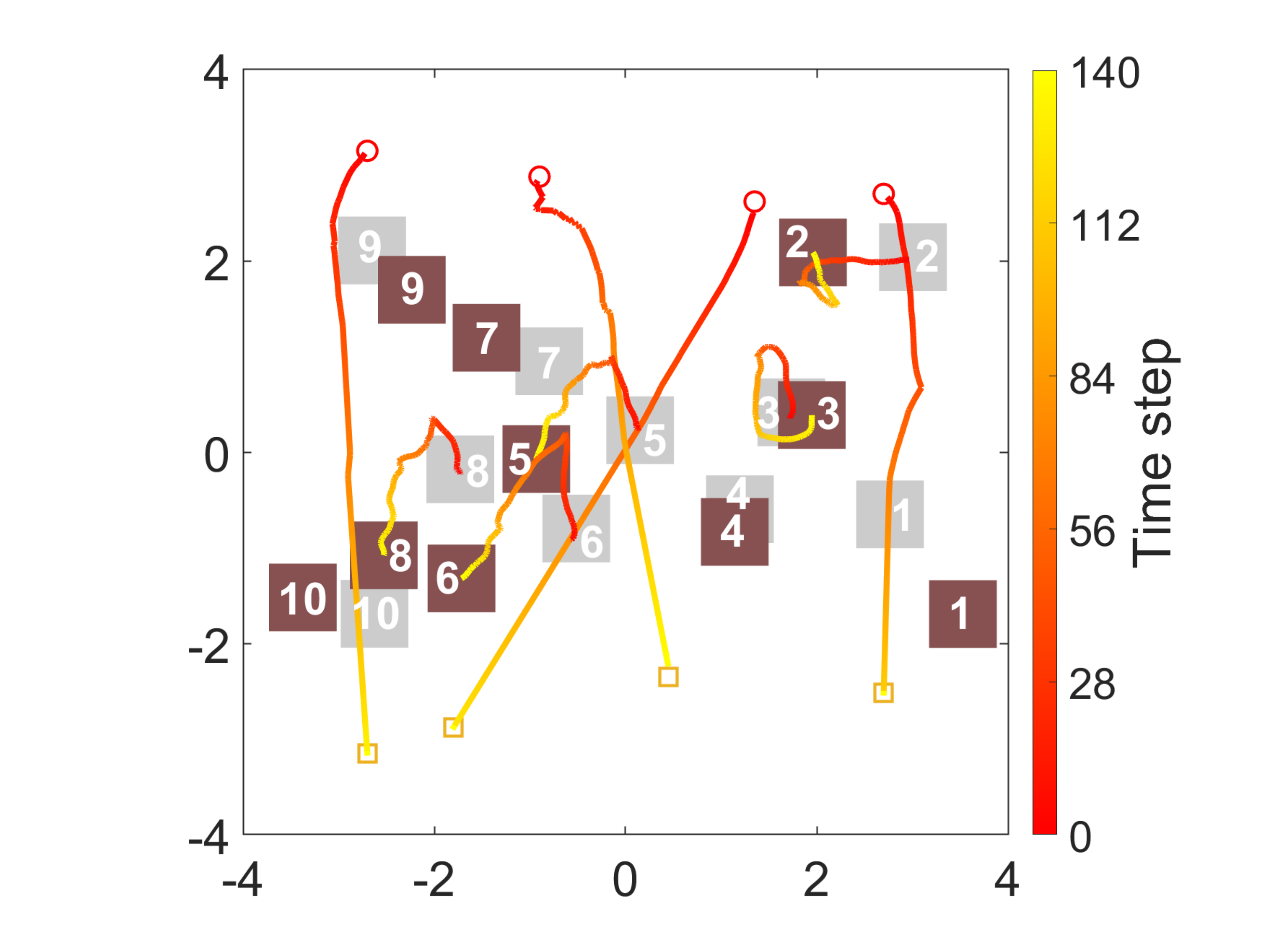}}
	\caption{(a-b) Performance of online environment optimization compared to the baselines. Results are averaged over $20$ trails and the shaded area shows the std. dev. The RL system is trained on $10$ obstacles and tested on $10$ to $18$ obstacles. (a) SPL ($1$ is best). (b) PCTSpeed. (c) Example of online environment optimization. Red circles are initial positions and yellow squares are destinations. Grey and brown rectangles are the obstacles before and after environment optimization, and are numbered for exposition. Red-to-yellow lines are trajectories of agents and $5$ example obstacles, and the color bar represents the time scale, showing that no agent-agent nor agent-obstacle collision occurs.}\label{fig:online}\vspace{-4mm}
\end{figure*}

\section{Experiments}\label{sec:experiments}

We evaluate the proposed approach in this section. First, we consider unprioritized environment optimization without constraints [Problem \ref{def:environmenProblem}]. Then, we consider prioritized environment optimization with constraints [Problem \ref{def:prioritizedEO}]. For both cases, we consider two navigation scenarios, one in which we perform offline optimization with discrete obstacle motion, and another in which we consider online optimization with continuous obstacle motion. The obstacles have rectangular bodies and the agents have circular bodies. The given agent trajectory planner $\pi_a$ is the Reciprocal Velocity Obstacles (RVO) method \cite{van2008reciprocal}.\footnote{The proposed environment optimization approach can be applied in the same manner, irrespective of the chosen agent navigation algorithm.} Two metrics are used: Success weighted by Path Length (SPL) \cite{anderson2018evaluation} and the percentage of the maximal speed (PCTSpeed). The former is a stringent measure combining the success rate and the path length, while the latter is the ratio of the average speed to the maximal one. Results are averaged over $20$ trials with random initial configurations. 

\begin{figure*}
	\centering
	\subfloat[]{\includegraphics[width=0.25\linewidth]{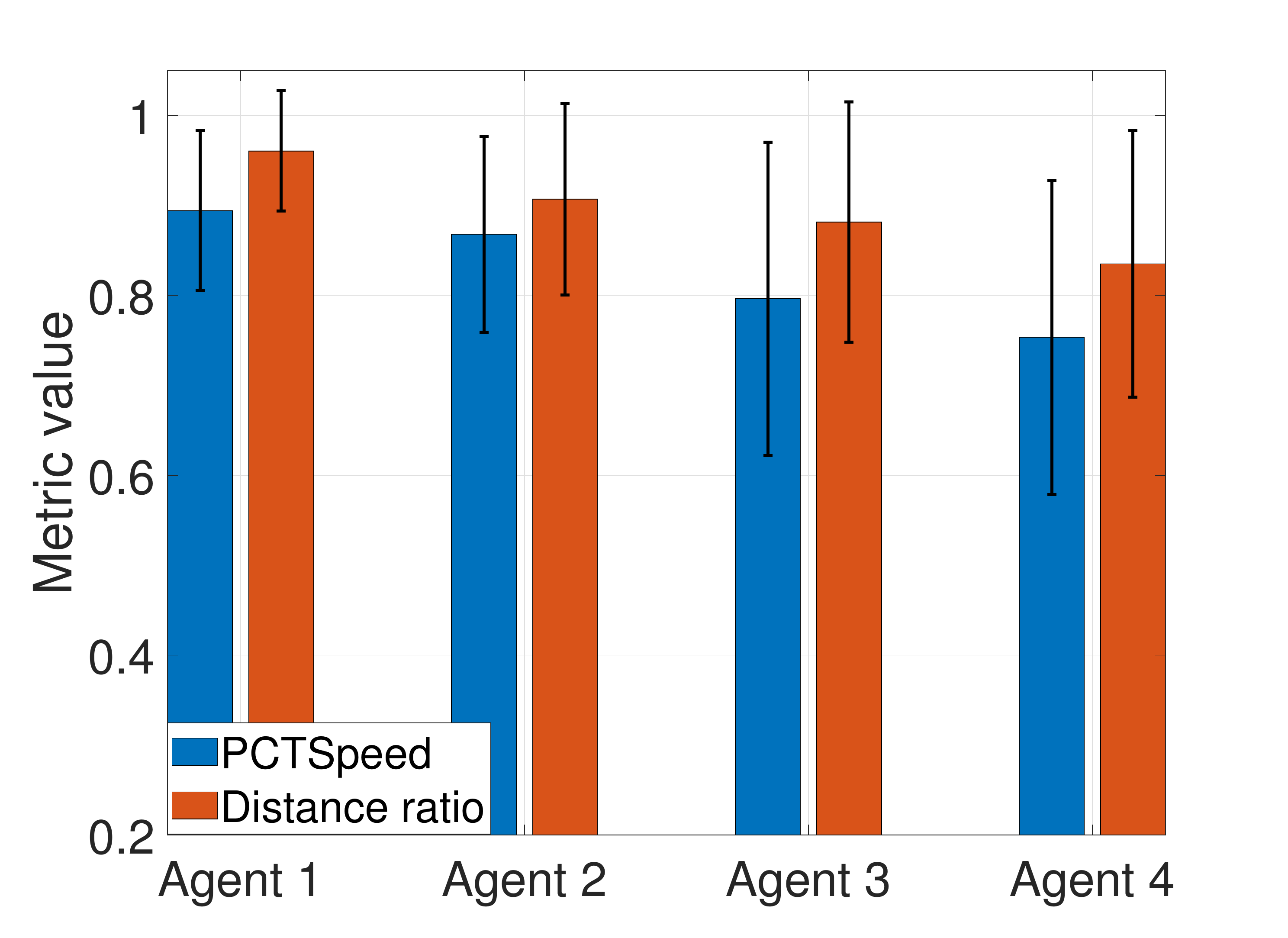}}
	\subfloat[]{\includegraphics[width=0.25\linewidth]{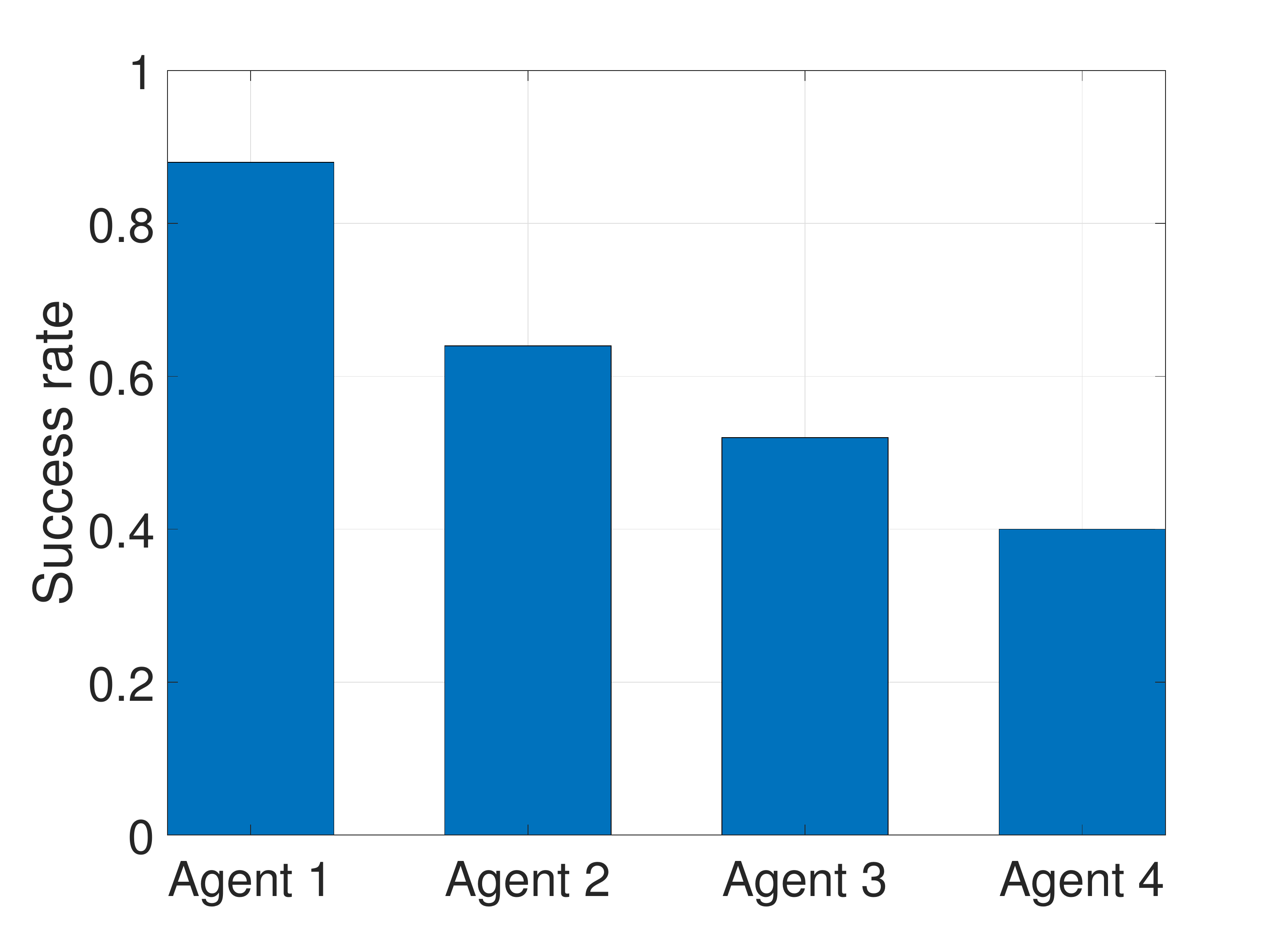}}
	\subfloat[]{\includegraphics[width=0.25\linewidth, height = 0.2\linewidth, trim = {2cm 0.7cm 3cm 0cm}, clip]{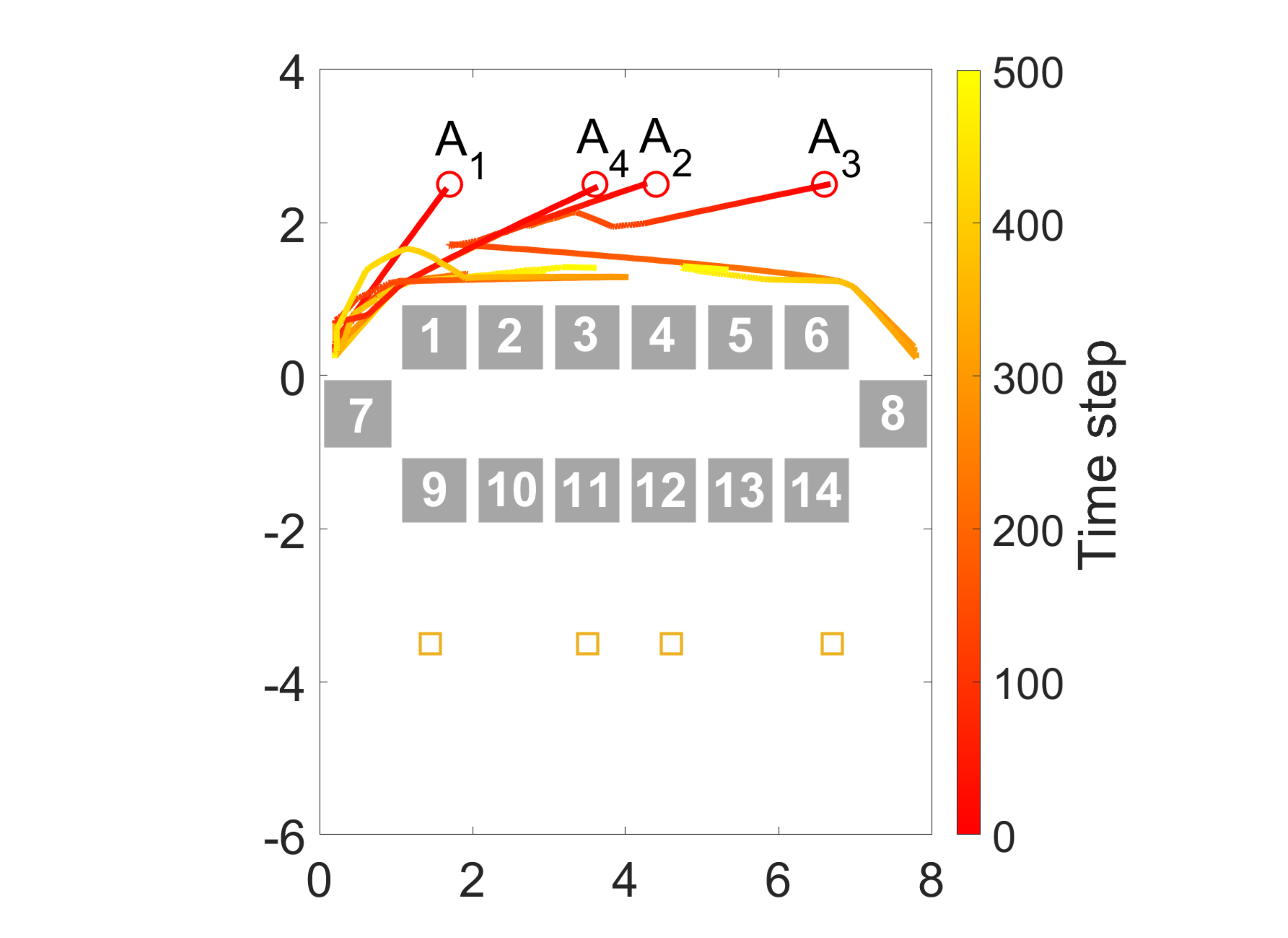}}
	\subfloat[]{\includegraphics[width=0.25\linewidth, height = 0.2\linewidth, trim = {2cm 0.7cm 3cm 0cm}, clip]{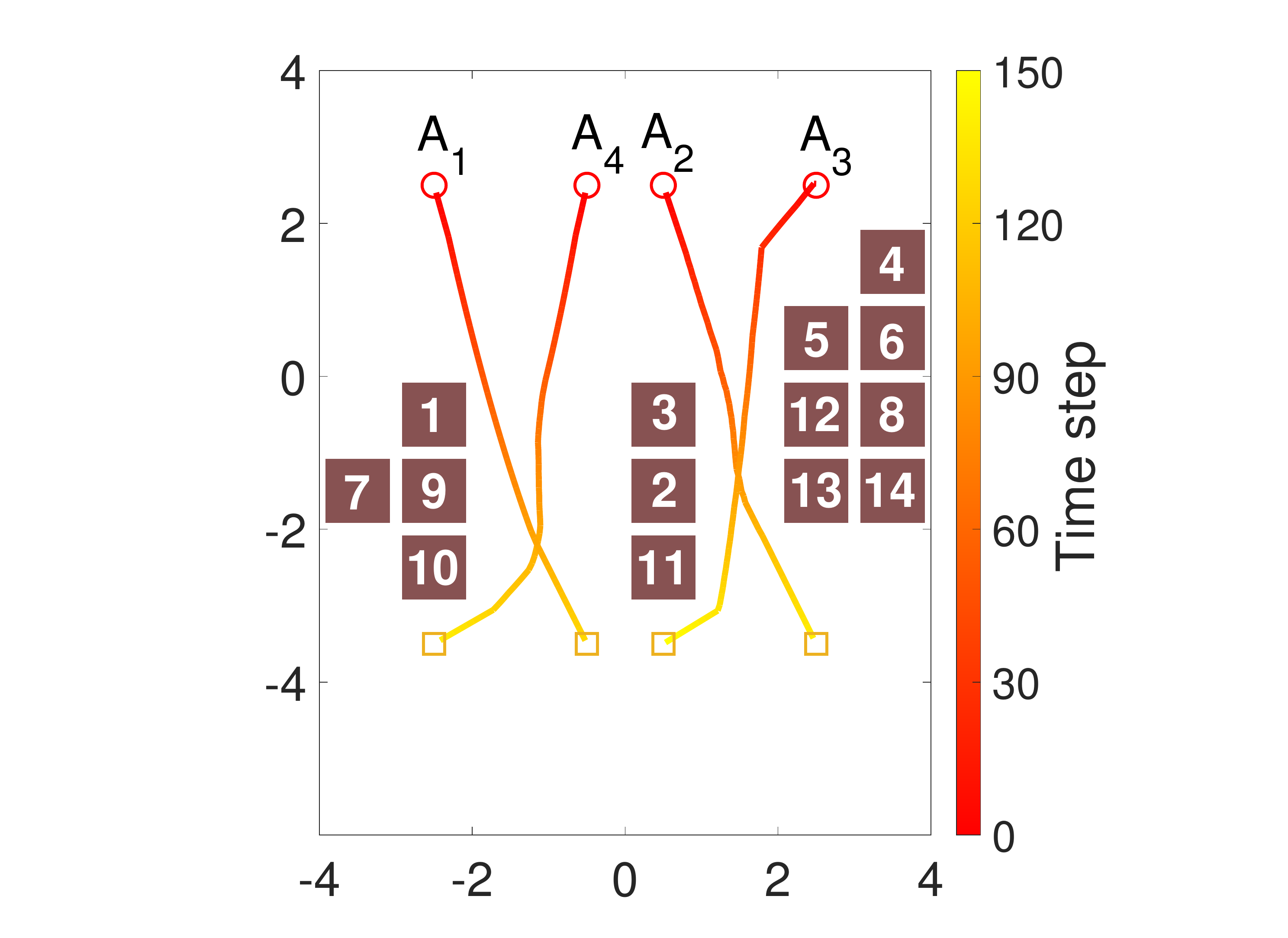}}
	\caption{(a-b) Performance of offline prioritized environment optimization with agent priorities $\rho_1 \ge \rho_2 \ge \rho_3 \ge \rho_4$. (a) PCTSpeeds and distance ratios of different agents on $14$ obstacles ($1$ is best). (b) Success rates of different agents on $24$ obstacles. (c-d) Example of offline prioritized environment optimization. Red circles are initial positions, yellow squares are destinations, and obstacles are numbered for exposition. 
		Color lines from red to yellow are agent trajectories and the color bar represents the time scale. Agent priorities $\{\rho_i\}_{i=1}^4$ reduce from $A_1$ to $A_4$, i.e., $\rho_1 \ge \cdots \ge \rho_4$. (c) Navigation failures in original environment. (d) Completeness with offline prioritized environment optimization. PCTSpeeds of agents are $0.98, 0.94, 0.88, 0.86$ and distance ratios are $0.99, 0.98, 0.95, 0.93$, i.e., agent performance strictly increases with the priority.}\label{fig:offline_prioritized}\vspace{-4mm}
\end{figure*}

\subsection{Unprioritized Environment Optimization}\label{subsec:unprioritizedEO}

We start with unprioritized environment optimization. 

\smallskip
\noindent\textbf{Offline discrete setting.} We consider a grid environment of size $8 \times 8$ with $10$ obstacles and $4$ agents, which are distributed randomly without overlap. The maximal agent / obstacle velocity is $0.05$ and the maximal time step is $500$.  

\smallskip
\noindent \emph{Setup.} The environment is modeled as an occupancy grid map. An agent's location is modeled by a one-hot in a matrix of the same dimension. At each step, the policy considers one of the obstacles and moves it to one of the adjacent grid cells. This repeats for $m$ obstacles, referred to as a round, and an episode ends if the maximal round has been reached.

\smallskip
\noindent \emph{Training.} The objective is to make agents reach destinations quickly while avoiding collision. The team reward is the sum of the PCTSpeed and the ratio of the shortest distance to the traveled distance, while the local reward is the collision penalty of individual obstacle [cf. \eqref{eq:obstacleRewardFunction}]. We parameterize the policy with a CNN of $4$ layers, each containing $25$ features with kernel size $2 \times 2$, and conduct training with PPO \cite{schulman2017proximal}.

\smallskip
\noindent \emph{Baseline.} Since exhaustive search methods are intractable for our problem, we develop a strong heuristic method to act as a baseline: At each step, one of the obstacles computes the shortest paths of all agents, checks whether it blocks any of these paths, and moves away randomly if so. This repeats for $m$ obstacles, referred to as a round, and the method ends if the maximal round is reached. 

\smallskip
\noindent \emph{Performance.} We train our model on $10$ obstacles and test on $10$ to $18$ obstacles, which varies obstacle density from $10\%$ to $30\%$. Figs. \ref{fig:offline}(a)-(b)  
show the results. The proposed approach consistently outperforms baselines with the highest SPL/PCTSpeed and the lowest variance. The heuristic method takes the second place and the original scenario (without any environment modification) performs worst. As we generalize to higher obstacle densities, all methods degrade as expected. However, our approach maintains a satisfactory performance due to the CNN's cabability for generalization. Figs. \ref{fig:offline}(c)-(d) show an example of how the proposed approach circumvents the dead-locks by optimizing the obstacle layout. Moreover, it improves the path efficiency such that all agents find collision-free trajectories close to their shortest paths.

\smallskip
\noindent \textbf{Online continuous setting.} We proceed to a continuous environment. The agents are initialized randomly in an arbitrarily defined starting region and towards destinations in an arbitrarily defined goal region.

\smallskip
\noindent \emph{Setup.} The agents and obstacles are modeled by positions $\{\bbp_{a,i}\}_{i=1}^n$, $\{\bbp_{o,j}\}_{j=1}^m$ and velocities $\{\bbv_{a,i}\}_{i=1}^n$, $\{\bbv_{o,j}\}_{j=1}^m$. At each step, each obstacle has a local policy that generates the desired velocity with neighborhood information and we integrate an acceleration-constrained mechanism for position changes. An episode ends if all agents reach destinations or the episode times out. The maximal acceleration is $0.05$, the communication radius is $2$ and the episode time is $500$. 

\smallskip
\noindent \emph{Training.} The team reward in \eqref{eq:obstacleRewardFunction} guides the agents to their destinations as quickly as possible and is defined as \vspace{-2mm}
\begin{align}
	r_{a,i}^{(t)} = \Big(\frac{\bbp_{i}^{(t)}-\bbd_i}{\|\bbp_{i}^{(t)}-\bbd_i\|_2} \cdot \frac{\bbv_{i}^{(t)}}{\|\bbv_{i}^{(t)}\|_2} \Big) \|\bbv_{i}^{(t)}\|_2
\end{align}
at time step $t$, which rewards fast movements towards the destination and penalizes moving away from it. The local reward is the collision penalty. We parameterize the policy with a single-layer GNN. The message aggregation function and feature update function are multi-layer perceptrons (MLPs), and we train the model using PPO.

\smallskip
\noindent \emph{Performance.} The results are shown in Figs. \ref{fig:online}(a)-(b). The proposed approach exhibits the best performance for both metrics and maintains a good performance for large scenarios. We attribute the latter to the fact that GNNs exploit topological information for feature extraction and are scalable to larger graphs. The heuristic method performs worse but comparably for small number of obstacles, while degrading with the increasing of obstacles. It is note-worthy that the heuristic method is centralized because it requires computing shortest paths of all agents, and hence is not applicable for online optimization and considered here as a benchmark value only for reference. Fig. \ref{fig:online}(c) shows the moving trajectories of agents and example obstacles. We see that obstacles make way for agents to facilitate navigation s.t. agents find trajectories close to their shortest paths.

\subsection{Prioritized Environment Optimization}\label{subsec:prioritizedEO}

We proceed to prioritized environment optimization with real-world constraints, where the agent priorities are set as $\bbrho = [2, 1, 0.5, 0.1]^\top$. The environment settings and implementation details are the same as the unprioritized environment optimization in Section \ref{subsec:unprioritizedEO}.

\smallskip
\noindent\textbf{Offline discrete setting.} We consider the constraint as the maximal round that can be performed by the policy $\pi_o$, i.e., the maximal number of grids each obstacle can move, which restricts the capacity of offline environment optimization. 
This constraint can be satisfied by setting the maximal length of an episode in reinforcement learning and thus, the primal-dual mechanism is not needed in this setting. We consider the maximal round as $8$ and measure the performance with the PCTSpeed and the distance ratio, where the latter is the ratio of the shortest distance to the agent's traveled distance. 

\begin{figure}
	\centering
	\subfloat[]{\includegraphics[width=0.5\linewidth]{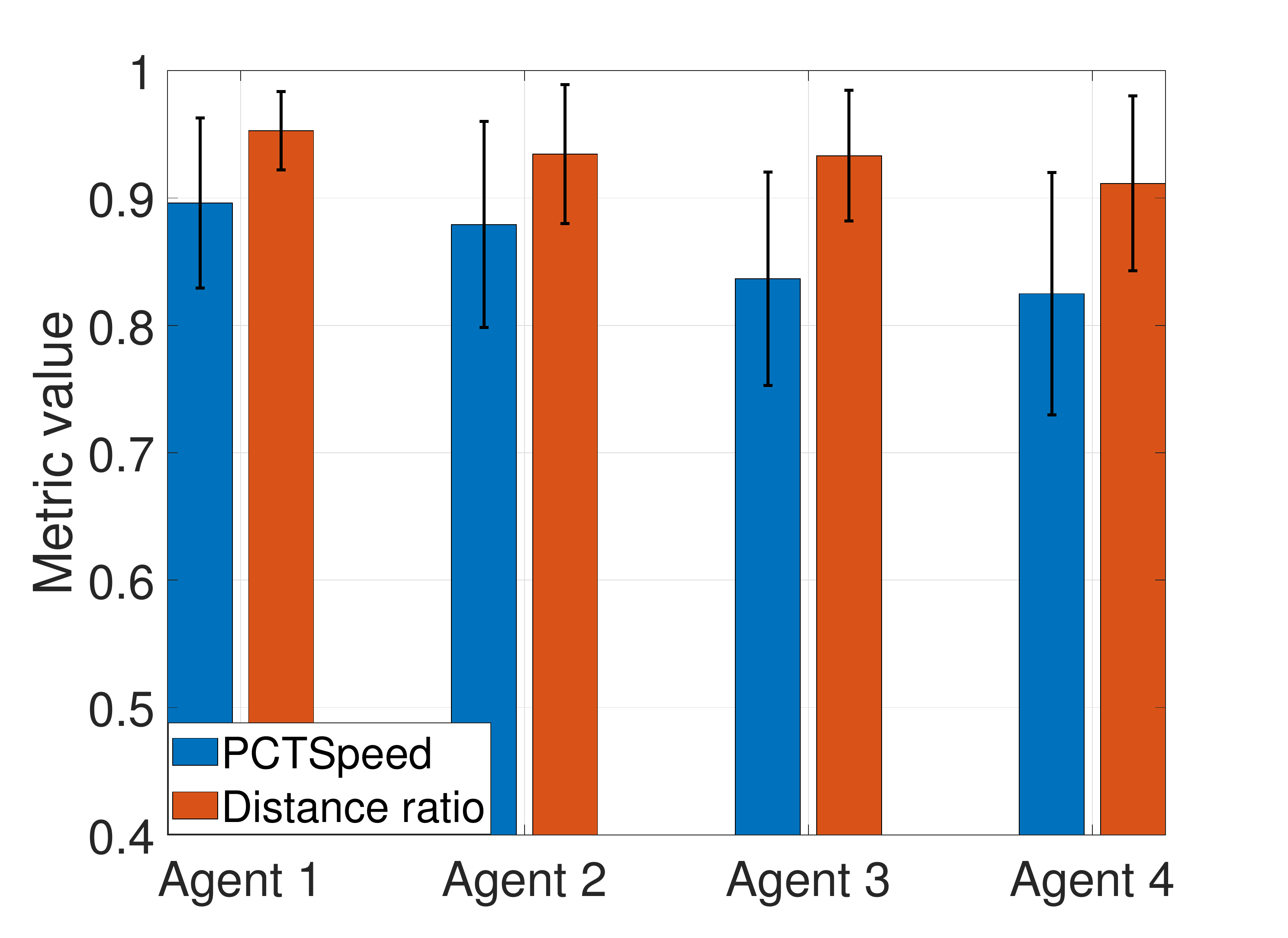}}
	\subfloat[]{\includegraphics[width=0.5\linewidth]{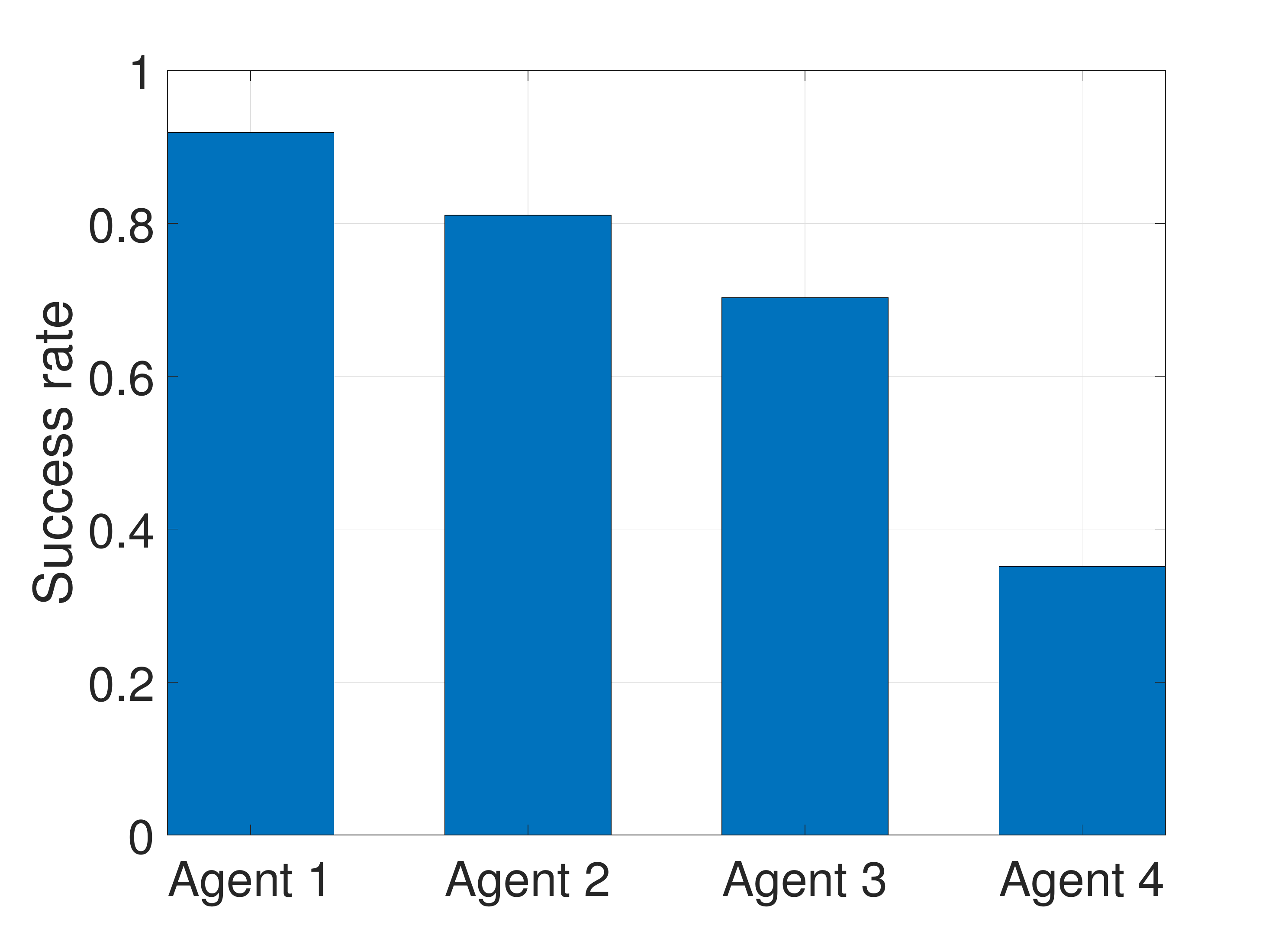}}
	\caption{Performance of online prioritized environment optimization with agent priorities $\rho_1 \ge \rho_2 \ge \rho_3 \ge \rho_4$. (a) PCTSpeeds and distance ratios of different agents on $14$ obstacles ($1$ is best). (b) Success rates of different agents within $150$ time steps.}\label{fig:online_prioritized}\vspace{-4mm}
\end{figure}

\begin{figure*}
	\centering
	\subfloat[]{\includegraphics[width=0.33\linewidth, height = 0.25\linewidth, trim = {0 0.7cm 3cm 0cm}, clip]{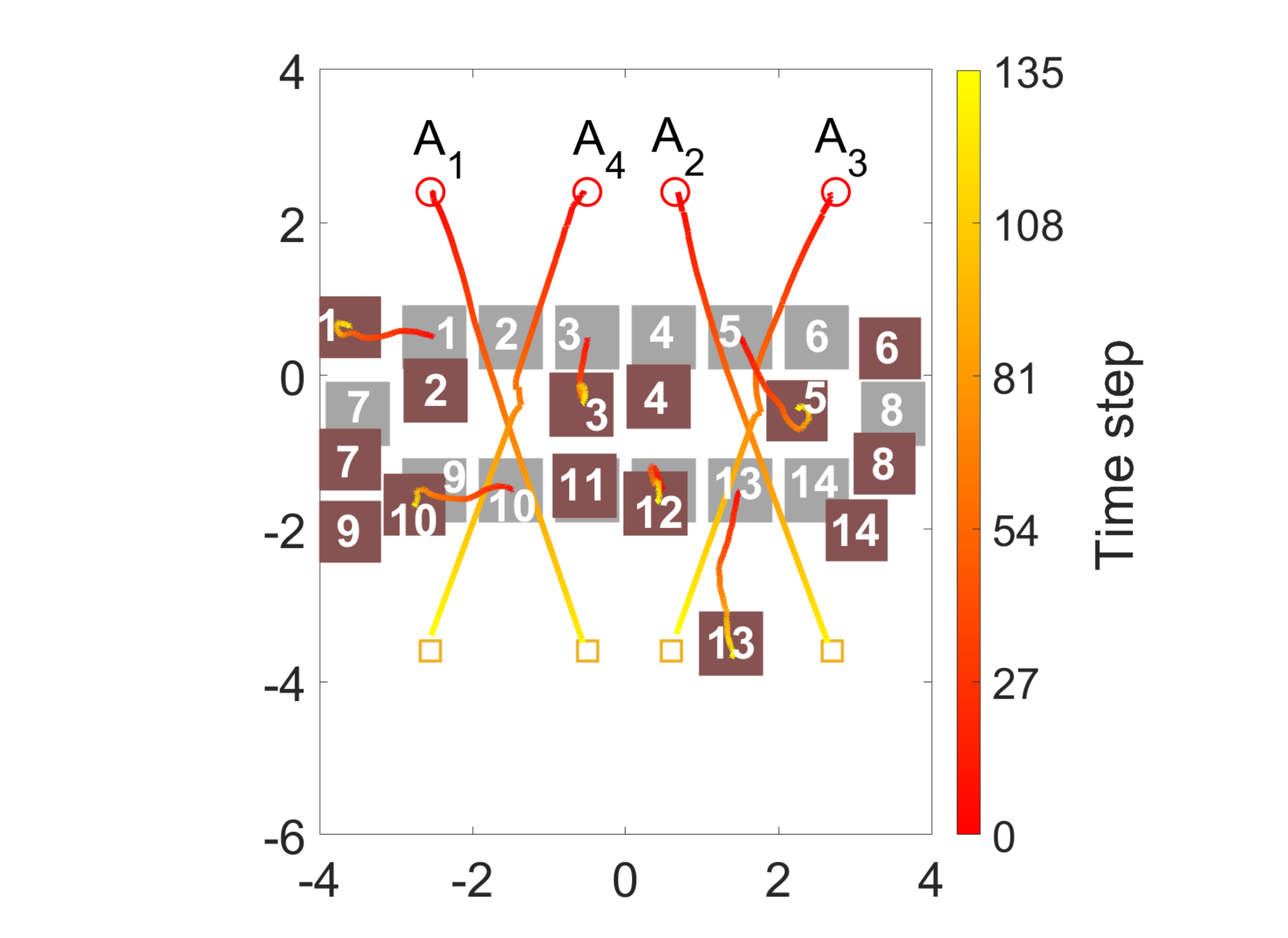}}
	\subfloat[]{\includegraphics[width=0.33\linewidth, height = 0.25\linewidth, trim = {0 0.7cm 3cm 0cm}, clip]{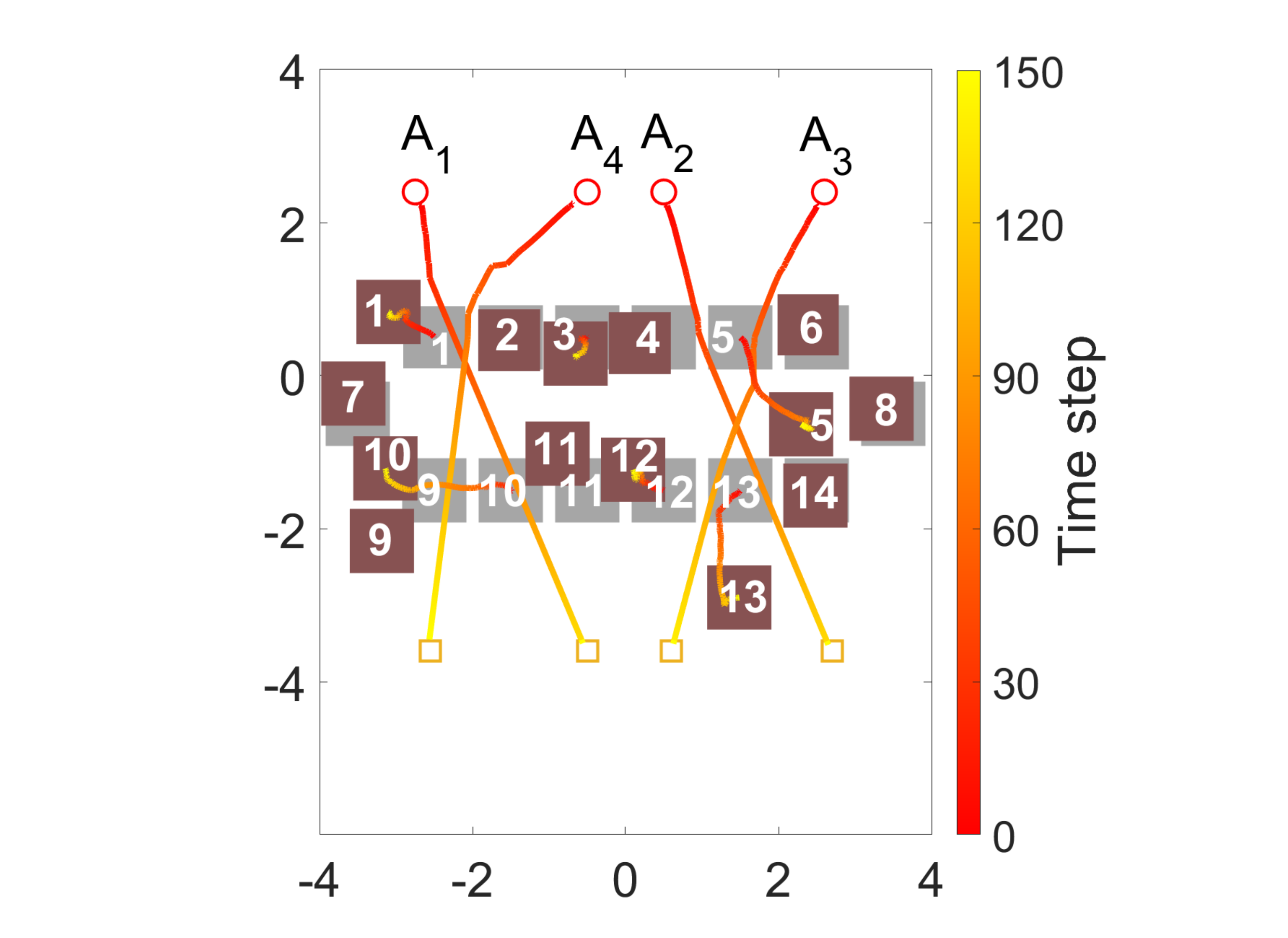}}
	\subfloat[]{\includegraphics[width=0.33\linewidth, height = 0.25\linewidth, trim = {0 0.7cm 3cm 0cm}, clip]{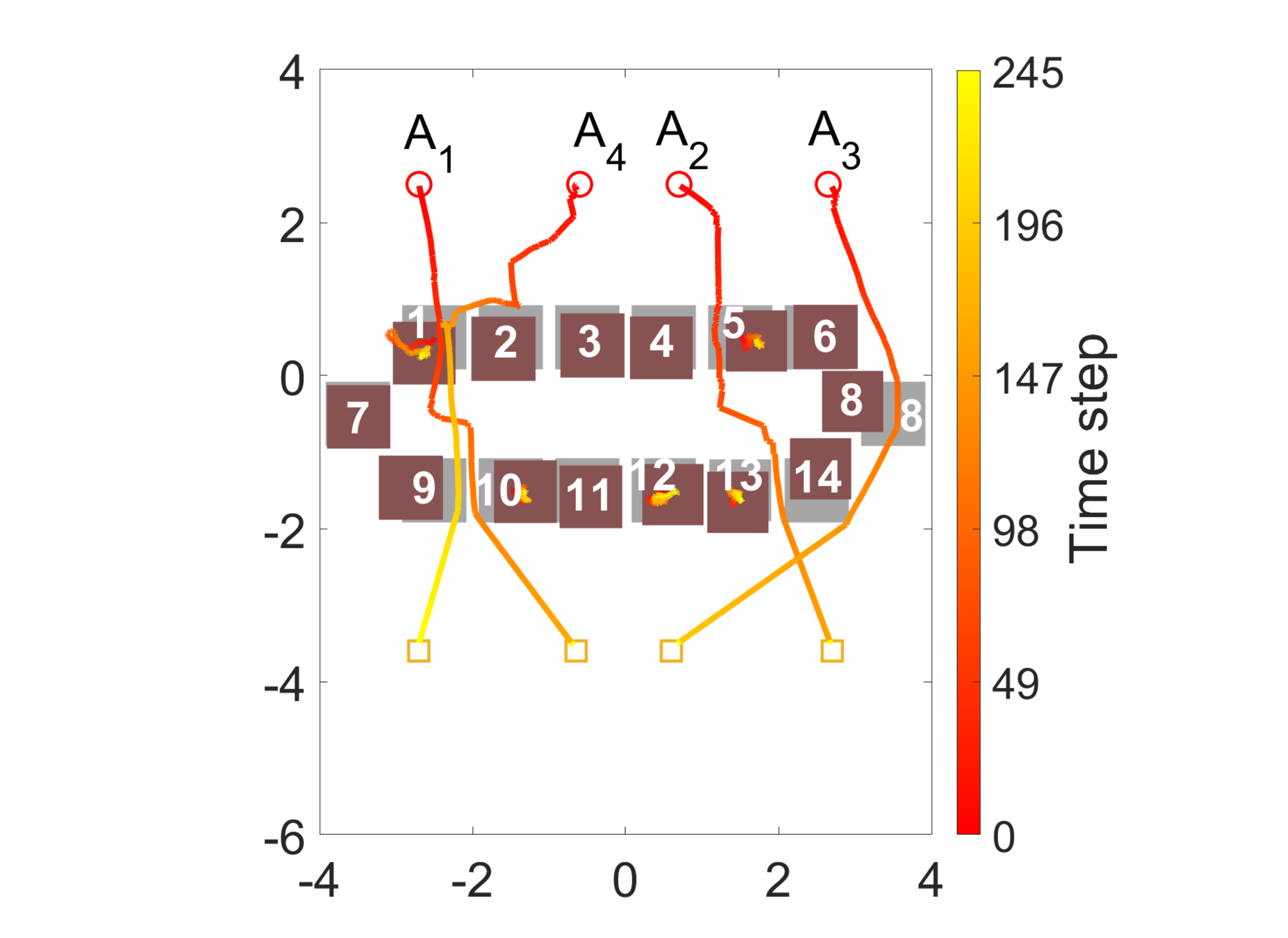}}
	\caption{Examples of online prioritized environment optimization. Red circles are initial positions and yellow squares are destinations. Grey and brown rectangles are obstacles before and after environment optimization, and are numbered for exposition. Color lines from red to yellow are agent trajectories and the color bar represents the time scale. Agent priorities $\{\rho_i\}_{i=1}^4$ reduce from $A_1$ to $A_4$, i.e., $\rho_1 \ge \cdots \ge \rho_4$. (a) Completeness with deviation distance constraint of $C_d=14$. PCTSpeeds of agents are $0.98, 0.97, 0.94, 0.91$ and distance ratios are $0.99, 0.99, 0.98, 0.97$. (b) Completeness with deviation distance constraint of $C_d=10$. PCTSpeeds of agents are $0.97, 0.95, 0.88, 0.84$ and distance ratios are $0.98, 0.98, 0.96, 0.94$. (c) Completeness with deviation distance constraint of $C_d=4$. PCTSpeeds of agents are $0.86, 0.82, 0.76, 0.63$ and distance ratios are $0.93, 0.92, 0.85, 0.82$. Agent performance strictly increases with the priority.}\label{fig:prioritized_example}\vspace{-4mm}
\end{figure*}

\begin{figure*}
\centering
\subfloat[]{\includegraphics[width=0.33\linewidth]{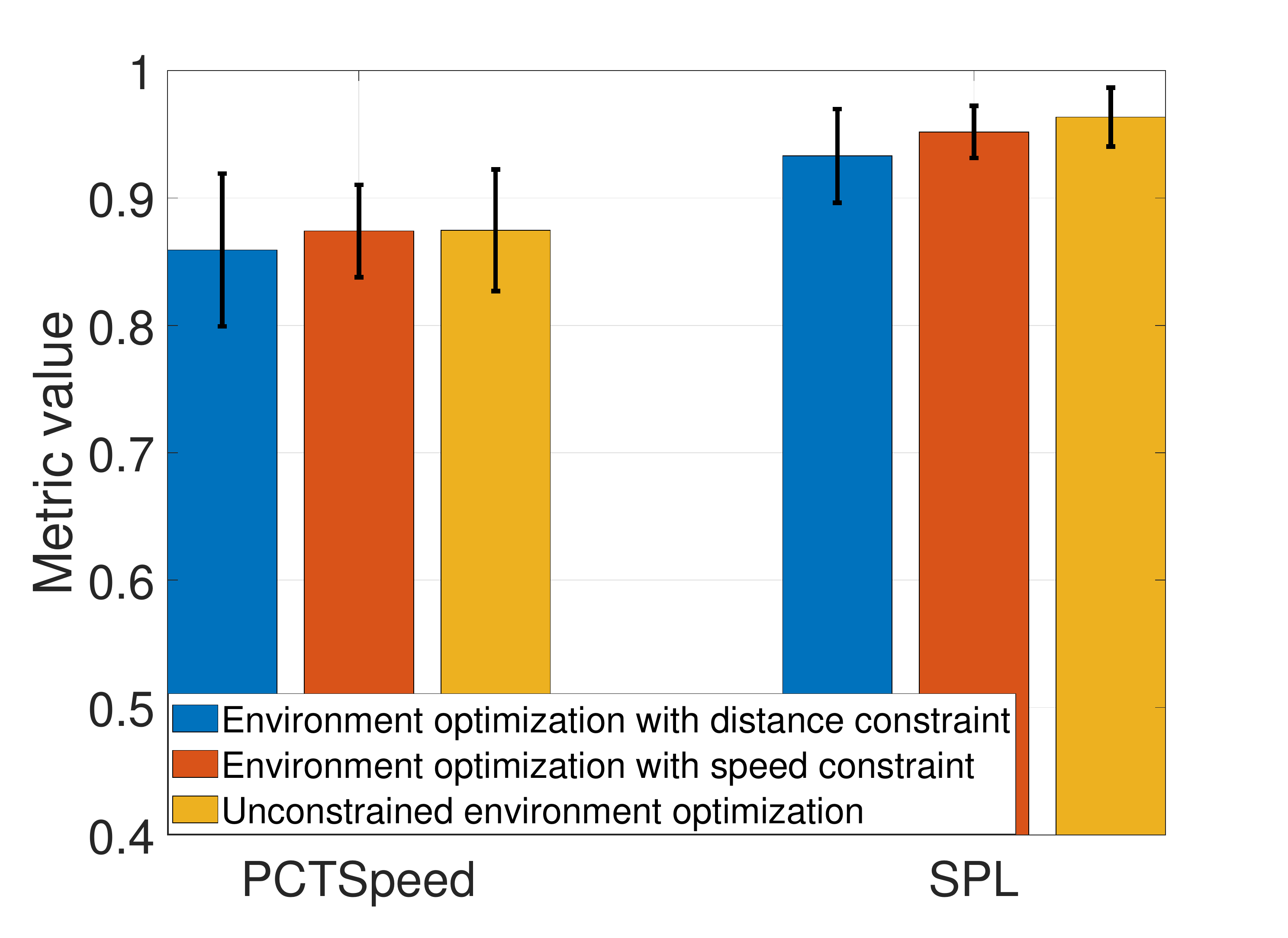}}
\subfloat[]{\includegraphics[width=0.33\linewidth]{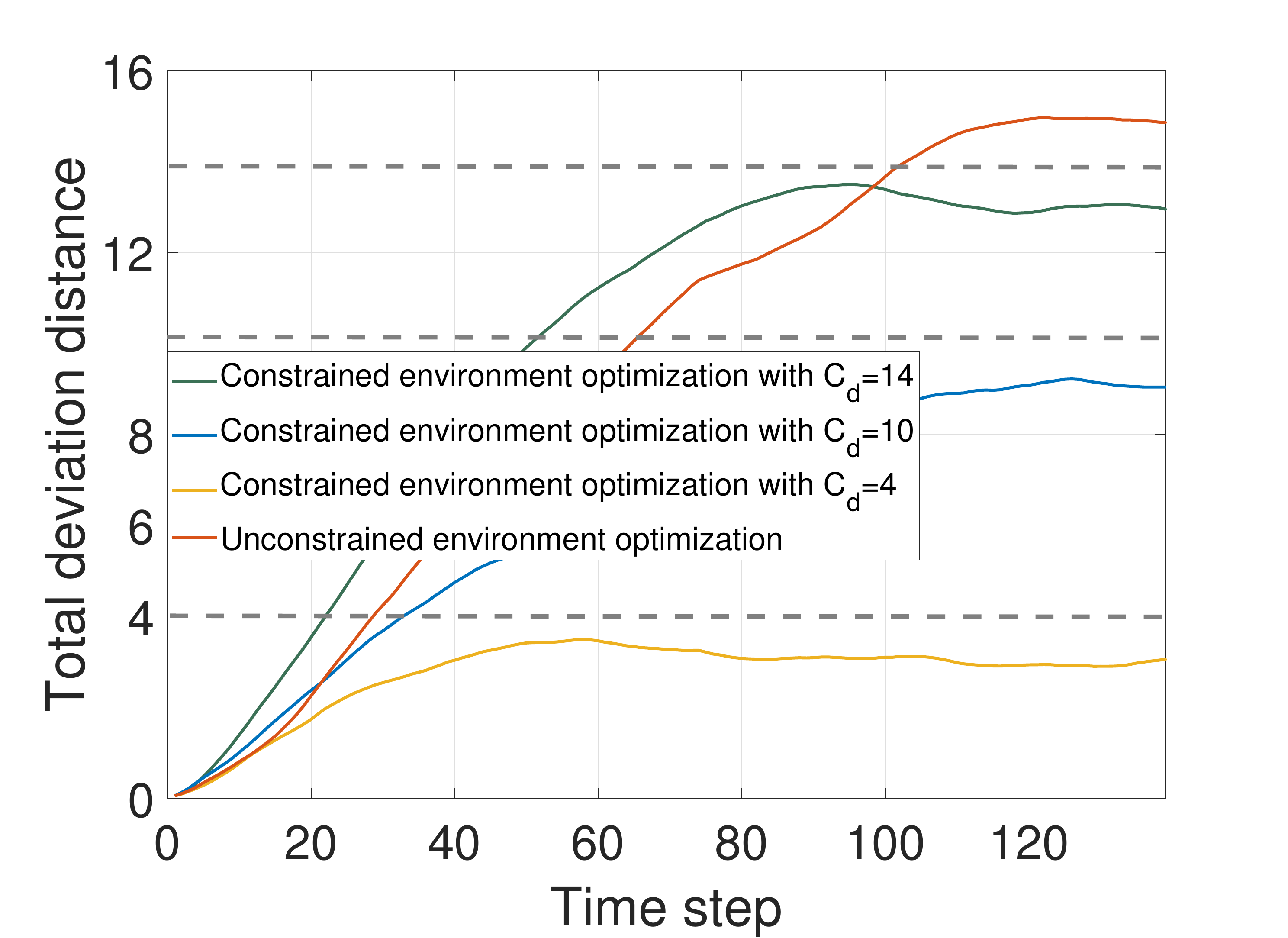}}
\subfloat[]{\includegraphics[width=0.33\linewidth]{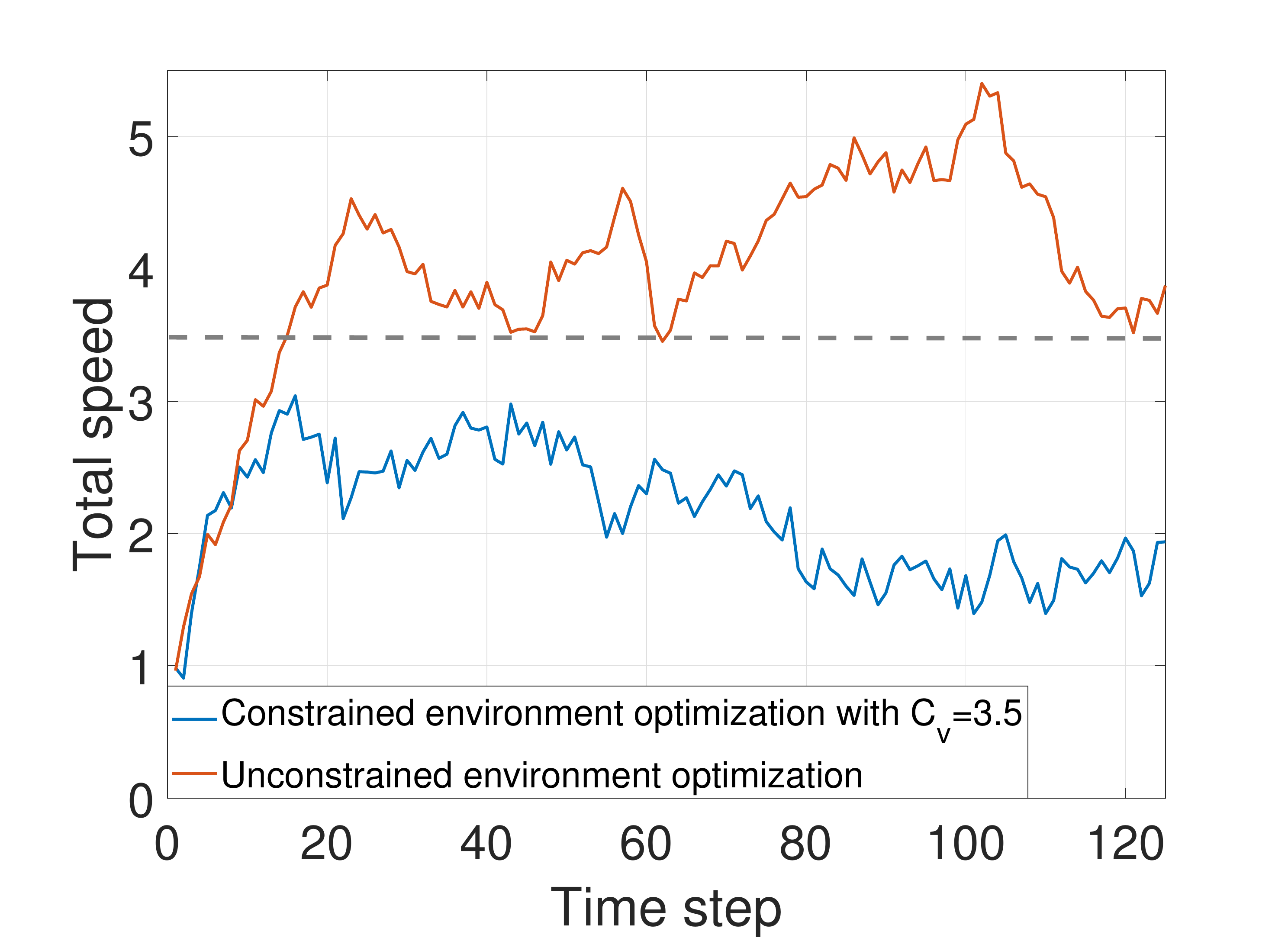}}
\caption{(a) Performance of online prioritized environment optimization with diferent constraints. (b) Deviation distance constraint of obstacles throughout the navigation procedure. (c) Speed constraint of obstacles throughout the navigation procedure.}\label{fig:comparison}\vspace{-4mm}
\end{figure*}

\smallskip
\noindent \emph{Performance.} We evaluate the proposed approach 
on $14$ obstacles and the results are shown in Fig. \ref{fig:offline_prioritized}(a). We see that agent $A_1$ with the highest priority exhibits the best performance with the highest PCTSpeed and distance ratio. The performance degrades as the agent priority decreases (from $A_1$ to $A_4$), which corresponds to theoretical findings in Theorem \ref{thm:prioritizedOfflineCompleteness}. Figs. \ref{fig:offline_prioritized}(c)-(d) display an example of how offline prioritized environment optimization optimizes the obstacle layout. It guarantees the success of all navigation tasks and improves agents' path efficiencies based on their priorities, i.e., it emphasizes the higher-priority agents ($A_1$ and $A_2$) over the lower-priority agents ($A_3$ and $A_4$).

\smallskip
\noindent \emph{Partial completeness.} We test the trained model on $24$ obstacles, where the obstacle region is dense and the obstacle-free area is limited. The success rates of all agents are zero without environment optimization. Fig. \ref{fig:offline_prioritized}(b) shows the results with offline prioritized environment optimization. Agent $A_1$ with the highest priority achieves the best success rate and the success rate decreases with the agent priority, which corroborates the partial completeness in Corollary \ref{coro:incompletePrioritizedOfflineEO}. We note that the success rate of $A_1$ is not one ($100 \%$) because (i) we train our model on $14$ obstacles but test it on $24$ obstacles and (ii) the proposed approach obtains a local not global solution, leading to inevitable performance degradation compared with theoretical analysis.

\smallskip
\noindent\textbf{Online continuous setting.} We consider the constraint as the total deviation distance of the obstacles away from their initial positions, i.e., $\sum_{j=1}^m \|\bbo_j^{(t)} - \bbo_j^{(0)}\|_2 \le C_d$ for all $t$, where $\bbo_j^{(t)}$ is the position of obstacle $j$ at time step $t$ and $\bbo_j^{(0)}$ is the initial position. It corresponds to practical scenarios where the obstacles are not allowed moving too far from their original positions. 

\smallskip
\noindent \emph{Performance.} We set $C_d = 10$ and show the 
results in Fig. \ref{fig:online_prioritized}(a). Agent $A_1$ with the highest priority performs best, and the performance degrades from $A_1$ to $A_4$ with the decreasing of agent priority, corroborating our analysis in Theorem \ref{thm:prioritizedOnlineCompleteness}. Fig. \ref{fig:comparison}(b) shows the constraint value as a function of time steps. The total deviation distance of the obstacles is smaller than the deviation bound throughout the navigation procedure, which validates the proposed primal-dual method. Fig. \ref{fig:prioritized_example}(b) shows the moving trajectories of agents and example obstacles in online prioritized environment optimization. The obstacles change positions to improve agents' performance and the higher-priority agents (e.g., $A_1$) are emphasized over the lower-priority agents (e.g., $A_4$). For example, given the deviation distance constraint, the obstacles create an (almost) shortest path for $A_1$ but not for $A_4$. 

\smallskip
\noindent \emph{Partial completeness.} We corroborate the partial completeness in Corollary \ref{coro:IncompletePrioritizedOnlineEO} by requiring that all agents arrive at destinations within $T_{\max}=150$ time steps. Fig. \ref{fig:online_prioritized}(b) shows that the success rate decreases from higher-priority agents to lower-priority ones. The success rate of $A_1$ is not one ($100 \%$) because (i) the obtained solution is local not global and (ii) we impose the constraint of deviation distance on the environment optimization. 

\smallskip
\noindent \emph{Constraint bound.} We test different constraint bounds $C_d=14, 10$ and $4$, the decreasing of which corresponds to the increasing of environment restrictions. 
Fig. \ref{fig:comparison}(b) shows that all variants satisfy the constraints, and Fig. \ref{fig:prioritized_example} displays three examples of agent and obstacle trajectories. We see that all variants guarantee the success of navigation tasks. The variant with the largest bound performs best with agent trajectories close to the shortest paths, where agent priorities do not play an important role given sufficient resources. The variant with the lowest bound suffers from performance degradation because of the strongest constraint. It puts more emphasis on the higher-priority agents ($A_1$ and $A_2$) while overlooking the lower-priority agents ($A_3$ and $A_4$), which corroborates theoretical analysis in Section \ref{sec:prioritizedEO}

%

\smallskip
\noindent \emph{Velocity constraint.} We show that the proposed approach can handle various constraints. Here, we test a constraint on the total speed of the obstacles, i.e., $\sum_{j=1}^m \|\bbv_j^{(t)}\|_2 \le C_v$ for all $t$. We set $C_v = 3.5$ and Fig. \ref{fig:comparison}(c) shows that the speed constraint is satisfied throughout the navigation procedure. The obstacle speed first increases to make way for the agents and then slows down to satisfy the constraint. Fig. \ref{fig:comparison}(a) compares the performance of environment optimization with different constraints to the unconstrained counterpart. The constrained variants achieve comparable performance (slightly worse), but saves the traveled distance and the velocity energy of the obstacles.  


\section{Conclusion} \label{sec:conclusion}

We proposed novel problems of unprioritized and prioritized environment optimization for multi-agent navigation, each of which contains offline and online variants. By conducting the completeness analysis, we provided conditions under which all navigation tasks are guaranteed success and identified the role played by agent priorities in environment optimization. We imposed constraints on the environment optimization corresponding to real-world restrictions, and formulated the latter as a constrained stochastic optimization problem. We leveraged model-free reinforcement learning together with a primal-dual mechanism to solve the problem. The former overcomes the challenge of explicitly modeling the relation between agents, environment and navigation performance, while the latter handles the constraints. By integrating different information processing architectures (e.g., CNNs and GNNs) for policy parameterization, the proposed approach can adapt to different implementation requirements. Numerical results corroborate theoretical findings, and show adaptability to various objectives and constraints.

\appendices 









\section{Proof of Theorem \ref{thm:offlineCompleteness}}\label{Appendix:Theorem1}

We prove the theorem as follows. First, we optimize $\Delta$ such that the environment is ``well-formed", i.e., any initial position in $\ccalS$ and destination in $\ccalD$ can be connected by a collision-free path. Then, we show the optimized environment guarantees the success of all navigation tasks. 

\smallskip
\noindent \textbf{Obstacle region optimization.} We first optimize $\Delta$ based on $\ccalS$ and $\ccalD$ to make the environment ``well-formed". To do so, we handle $\ccalS$, $\ccalD$ and the other space $\ccalE \setminus (\ccalS \bigcup \ccalD)$ separately. 

\emph{(i)} From Assumption \ref{as:initialDistribution}, the initial positions $\{\bbs_i\}_{i=1}^n$ in $\ccalS$ are distributed such that $d(\bbs_i,\bbs_j) \ge 2 \hat{r}$ and $d(\bbs_i,\partial \ccalS) \ge 2 \hat{r}$. Thus, for any $\bbs_i$, there exists a boundary point $\partial \bbs_i \in \partial \ccalS$ and a path $\bbp_{\bbs_i}^{\partial \bbs}$ connecting $\bbs_i$ and $\partial \bbs_i$ that is collision-free with respect to the other initial positions. 

\emph{(ii)} Similar result applies to the destinations $\{\bbd_i\}_{i=1}^n$ in $\ccalD$, i.e., for any $\bbd_i$, there exists a boundary point $\partial \bbd_i \in \partial \ccalD$ and a path $\bbp_{\partial \bbd_i}^{\bbd_i}$ connecting $\partial \bbd_i$ and $\bbd_i$ that is collision-free with respect to the other destinations.

\emph{(iii)} Consider $\partial \bbs_i$ and $\partial \bbd_i$ for agent $A_i$. The shortest path $\bbp_{\partial \bbs_i}^{\partial \bbd_i}$ that connects them is the straight path, the area of which is bounded as 
\begin{align}
	|\bbp_{\partial \bbs_i}^{\partial \bbd_i}| 
	\le 2 d_{\max}\hat{r}
\end{align}
because $d_{\max}$ is the maximal distance between $\ccalS$ and $\ccalD$. From $|\ccalE \setminus (\Delta \cup \ccalS \cup \ccalD)| \ge 2 n d_{\max} \hat{r}$ in \eqref{eq:offlineCompletenessCondition}, the area of the obstacle-free space in $\ccalE \setminus (\ccalS \bigcup \ccalD)$ is larger than $2 n d_{\max} \hat{r}$. Thus, we can always optimize $\Delta$ to $\Delta^*$ such that the path $\bbp_{\partial \bbs_i}^{\partial \bbd_i}$ is obstacle-free. If $\bbp_{\partial \bbs_i}^{\partial \bbd_i}$ dose not overlap with $\ccalS$ and $\ccalD$, we can connect $\partial \bbs_i$ and $\partial \bbd_i$ with $\bbp_{\partial \bbs_i}^{\partial \bbd_i}$ directly. If $\bbp_{\partial \bbs_i}^{\partial \bbd_i}$ passes through $\ccalS$ for $K$ times, let $\bbs_{i,e}^{(k)}$ and $\bbs_{i,l}^{(k)}$ be the entering and leaving positions of $\bbp_{\partial \bbs_i}^{\partial \bbd_i}$ on $\ccalS$ at $k$th pass for $k=1,...,K$ with $\bbs_{i,l}^{(0)} = \partial \bbs_i$ the initial leaving position. First, we can connect $\bbs_{i,l}^{(k-1)}$ and $\bbs_{i,e}^{(k)}$ by $\bbp_{\partial \bbs_i}^{\partial \bbd_i}$ because $\bbp_{\partial \bbs_i}^{\partial \bbd_i}$ is obstacle-free. Then, there exists a collision-free path $\bbp_i^{(k)}$ inside $\ccalS$ that connects $\bbs_{i,e}^{(k)}$ and $\bbs_{i,l}^{(k)}$ as described in (i). Same result applies to that $\bbp_{\partial \bbs_i}^{\partial \bbd_i}$ passes through $\ccalD$. Therefore, we can connect $\partial \bbs_i$ and $\partial \bbd_i$ with $\bbp_{\partial \bbs_i}^{\partial \bbd_i}$ and $\{\bbp_i^{(k)}\}_{k=1}^K$. 

By concatenating $\bbp_{\bbs_i}^{\partial \bbs_i}$, $\bbp_{\partial \bbs_i}^{\partial \bbd_i}$, $\{\bbp_i^{(k)}\}_{k=1}^K$ and $\bbp_{\partial \bbd_i}^{\bbd_i}$, we can establish the path $\bbp_{\bbs_i}^{\bbd_i}$ connecting $\bbs_i$ to $\bbd_i$ that is collision-free w.r.t. the other initial positions, destinations and the optimized obstacle region $\Delta^*$ for $i=1,...,n$, i.e., the optimized environment is ``well-formed".

\smallskip
\noindent \textbf{Completeness.} From Assumption \ref{as:trajectory} and the fact that the optimized environment is ``well-formed", Theorem 4 in \cite{vcap2015complete} shows that all navigation tasks will be carried out successfully without collision. Therefore, there exists an offline environment optimization scheme that guarantees the success of all navigation tasks completing the proof.	

\section{Proof of Theorem \ref{thm:onlineCompleteness}}\label{proof:Thm2}

We prove the theorem as follows. First, we separate the navigation procedure into $H$ time slices. Then, we optimize the obstacle region based on the agent positions at each time slice and show the completeness of individual time-sliced multi-agent navigation. Lastly, we show the completeness of the entire multi-agent navigation by concatenating individual time slices and complete the proof by limiting the number of time slices to the infinity, i.e., $H \to \infty$. 	

\smallskip
\noindent \textbf{Navigation procedure separation.} Let $T$ be the maximal operation time of trajectories $\{\bbp_i\}_{i=1}^n$ and $\{[(h-1)T/H, hT/H]\}_{h=1}^{H}$ the separated time slices. This yields intermediate positions $\big\{\bbp_i(hT/H)\big\}_{h=0}^{H}$ 
with $\bbp_i(0)=\bbs_i$ and $\bbp_i(T)=\bbd_i$ for $i=1,...,n$. We can re-formulate the navigation task into $H$ sub-navigation tasks, where the $h$th sub-navigation task of agent $A_i$ is from $\bbp_i((h-1)T/H)$ to $\bbp_i(hT/H)$ and the operation time of the sub-navigation task is $\delta t = T/H$. At each time slice, we first change the obstacle region based on the corresponding sub-navigation task and then navigate the agents until the next time slice. 

\smallskip
\noindent \textbf{Obstacle region optimization.} We consider each sub-navigation task separately and start from the $1$st one. Assume the obstacle region $\Delta$ satisfies 
\begin{align}\label{eq:proofTheorem21}
	|\ccalE \setminus (\Delta \cup \ccalS \cup \ccalD)| > 2 n \hat{r} \|\hat{\bbv}\|_2 \delta t.
\end{align}
For the $1$st sub-navigation task, the starting region is $\ccalS^{(1)} = \bigcup_{i=1,\ldots,n} B(\bbp_i(0), r_i) = \ccalS$ and the destination region is $\ccalD^{(1)} = \bigcup_{i=1,\ldots,n} B(\bbp_i(T/H), r_i)$. We optimize $\Delta$ based on $\ccalS^{(1)}$, $\ccalD^{(1)}$ and show the completeness of the $1$st sub-navigation task, which consists of two steps. First, we change $\Delta$ to $\widetilde{\Delta}$ such that $|\Delta| = |\widetilde{\Delta}|$ and $\widetilde{\Delta} \subset \ccalE \setminus (\ccalS^{(1)} \cup \ccalD^{(1)})$. This can be completed as follows. From the condition $\Delta \subset \ccalE \setminus (\ccalS \cup \ccalD)$ and $\ccalS^{(1)} = \ccalS$, there is no overlap between $\Delta$ and $\ccalS^{(1)}$. For any overlap region $\Delta \bigcap \ccalD^{(1)}$, we can change it to the obstacle-free region in $\ccalD$ because $\Delta \subset \ccalE \setminus (\ccalS \cup \ccalD)$ and $|\ccalD^{(1)}| = |\ccalD|$, and keep the other region in $\Delta$ unchanged. The resulting $\widetilde{\Delta}$ satisfies $|\widetilde{\Delta}| = |\Delta|$ and $\widetilde{\Delta} \subset \ccalE \setminus (\ccalS^{(1)} \cup \ccalD^{(1)})$. The changed area from $\Delta$ to $\widetilde{\Delta}$ is bounded by $|\ccalD^{(1)}|\le n \pi \hat{r}^2$. Second, we change $\widetilde{\Delta}$ to $\Delta^{(1)}$ such that the environment is ``well-formed" w.r.t. the $1$st sub-navigation task. The initial position $\bbp_i(0)$ and the destination $\bbp_i(H/T)$ can be connected by a path $\bbp_i^{(1)}$ that follows the trajectory $\bbp_i$. Since $\|\hat{\bbv}\|_2$ is the maximal speed and $\delta t$ is the operation time, the area of $\bbp_i^{(1)}$ is bounded by $2 \hat{r} \| \hat{\bbv} \|_2 \delta t$. Since this holds for all $i=1,...,n$, we have $\sum_{i=1}^n |\bbp_i^{(1)}| \le 2 n \hat{r} \| \hat{\bbv} \|_2 \delta t$. From \eqref{eq:proofTheorem21}, $|\ccalS^{(1)}| = |\ccalS|$, $|\ccalD^{(1)}| = |\ccalD|$ and $|\widetilde{\Delta}| = |\Delta|$, we have
\begin{align}\label{eq:proofTheorem23}
	|\ccalE \setminus (\widetilde{\Delta} \cup \ccalS^{(1)} \cup \ccalD^{(1)})| > 2 n \hat{r} \|\hat{\bbv}\|_2 \delta t.
\end{align}		
This implies that the obstacle-free area in $\ccalE$ is larger than the area of $n$ paths $\{\bbp_i^{(1)}\}_{i=1}^n$. Following the proof of Theorem \ref{thm:offlineCompleteness}, we can optimize $\widetilde{\Delta}$ to $\Delta^{(1)}$ to guarantee the success of the $1$st sub-navigation task. The changed area from $\widetilde{\Delta}$ to $\Delta^{(1)}$ is bounded by $2 n \hat{r} \|\hat{\bbv}\|_2 \delta t - n \pi \hat{r}^2$ because the initial positions $\{\bbp_i(0)\}_{i=1}^n$ and destinations $\{\bbp_i(T/H)\}_{i=1}^n$ in $\{\bbp_i^{(1)}\}_{i=1}^n$ are collision-free from the first step, which dose not require any further change of the obstacle region. The total changed area from $\Delta$ to $\Delta^{(1)}$ can be bounded as
\begin{align}\label{eq:proofTheorem24}
	\frac{\big|(\Delta \bigcup \Delta^{(1)}) \setminus (\Delta \bigcap \Delta^{(1)})\big|}{2} \le 2 n \hat{r} \|\hat{\bbv}\|_2 \delta t.
\end{align}

From \eqref{eq:proofTheorem23}, $|\Delta^{(1)}| = |\widetilde{\Delta}|$ and $\Delta^{(1)} \subset \ccalE \setminus (\ccalS^{(1)} \cup \ccalD^{(1)})$, the optimized $\Delta^{(1)}$ satisfies $|\ccalE \setminus (\Delta^{(1)} \cup \ccalS^{(1)} \cup \ccalD^{(1)})| \ge 2 n \hat{r} \|\hat{\bbv}\|_2 \delta t$, which recovers the assumption in \eqref{eq:proofTheorem21}. Therefore, we can repeat the above process and guarantee the success of $H$ sub-navigation tasks. The entire navigation task is guaranteed success by concatenating these sub-tasks. 

\smallskip
\noindent \textbf{Completeness.} When $H \to \infty$, we have $\delta t \to 0$. Since the environment optimization time is same as the agent operation time at each sub-navigation task, the obstacle region and the agents can be considered taking actions simultaneously when $\delta t \to 0$. The initial environment condition in \eqref{eq:proofTheorem21} becomes
\begin{align}\label{eq:proofTheorem27}
	\lim_{\delta t \to 0}| \ccalE \setminus (\Delta \cup \ccalS \cup \ccalD)| > 2 n \hat{r} \|\hat{\bbv}\|_2 \delta t \to 0.
\end{align}	 
which is satisfied from the condition $\ccalW \setminus (\Delta \cup \ccalS \cup \ccalD) \ne \emptyset$. The changed area of the obstacle region in \eqref{eq:proofTheorem24} becomes
\begin{align}\label{eq:proofTheorem28}
	\!\lim_{\delta t \to 0}\!\frac{\big|(\Delta^{(h)} \!\bigcup\! \Delta^{(h+1)}) \!\setminus\! (\Delta^{(h)} \!\bigcap\! \Delta^{(h+1)})\big|}{2 \delta t} \!\le\! 2 n \hat{r} \|\hat{\bbv}\|_2.
\end{align}
That is, if the capacity of the online environment optimization is stronger than $2 n \hat{r} \|\hat{\bbv}\|_2$, i.e., $\dot{\Delta} \ge 2 n \hat{r} \| \hat{\bbv} \|_2$, the navigation task can be carried out successfully without collision. Therefore, there exists an online environment optimization scheme that guarantees the success of all navigation tasks.

\section{Proof of Theorem \ref{thm:prioritizedOfflineCompleteness}}\label{proof:Theorem3}

We start by considering agent $A_1$ with the highest priority. From the condition \eqref{eq:prioritizedOfflineCompletenessCondition}, we have
\begin{align}\label{eq:proofThm31}
	|\ccalE \!\setminus\! (\Delta \!\cup\! \ccalS \!\cup\! \ccalD)| \!\ge\! 2 d_{A_1} \hat{r}.
\end{align}		
Following the proof of Theorem \ref{thm:offlineCompleteness}, we can optimize the obstacle region $\Delta$ to $\Delta^{(1)}$ such that for the initial position $\bbs_1$ and destination $\bbd_1$ of agent $A_1$, there exists a collision-free path $\bbp_1$ that connects $\bbs_1$ and $\bbd_1$.

We then consider the second agent $A_2$. From Assumption \ref{as:initialDistribution}, the initial positions $\{\bbs_i\}_{i=1}^n$ in $\ccalS$ are distributed such that $d(\bbs_i,\bbs_j) \ge 2 \hat{r}$, $d(\bbs_i,\partial \ccalS) \ge 2 \hat{r}$ and $\bbs_1$, $\bbs_2$ are distributed in the starting components $\ccalS_{A_1}$, $\ccalS_{A_2}$. If $\ccalS_{A_1} = \ccalS_{A_2}$, there exists a path $\bbp_{\bbs_2}^{\bbs_1}$ in $\ccalS_{A_1}$ without changing the obstacle region. If $\ccalS_{A_1} \ne \ccalS_{A_2}$, there exist boundary points $\partial \bbs_1 \in \partial \ccalS_{A_1}$, $\partial \bbs_2 \in \partial \ccalS_{A_2}$ and paths $\bbp_{\bbs_1}^{\partial \bbs_1}$, $\bbp_{\bbs_2}^{\partial \bbs_2}$ that connect $\bbs_1$, $\partial \bbs_1$ and $\bbs_2$, $\partial \bbs_2$, respectively, and are collision-free with respect to the other initial positions. For the boundary points $\partial \bbs_2$ and $\partial \bbs_1$, the shortest path $\bbp_{\partial \bbs_2}^{\partial \bbs_1}$ that connects them is the straight path, the area of which is bounded as 
\begin{align}\label{eq:proofThm32}
	|\bbp_{\partial \bbs_2}^{\partial \bbs_1}| 
	\le 2 d_{\max}(\ccalS_{A_2}, \ccalS_{A_1})\hat{r}
\end{align}
where $d_{\max}(\ccalS_{A_2}, \ccalS_{A_1})$ is the distance between $\ccalS_{A_1}$, $\ccalS_{A_2}$ and $d_{A_2, \ccalS} = d_{\max}(\ccalS_{A_2}, \ccalS_{A_1})$ by definition. Similar result applies to the destinations $\bbd_{1}$, $\bbd_{2}$ in $\ccalD_{A_1}$, $\ccalD_{A_2}$. That is, if $\ccalD_{A_1} = \ccalD_{A_2}$, there exists a path $\bbp_{\bbd_1}^{\bbd_2}$ in $\ccalD_{A_1}$ without changing the obstacle region. If $\ccalD_{A_1} \ne \ccalD_{A_2}$, there exist boundary points $\partial \bbd_1 \in \partial \ccalD_{A_1}$, $\partial \bbd_2 \in \partial \ccalD_{A_2}$ and paths $\bbp_{\partial \bbd_1}^{\bbd_1}$, $\bbp_{\partial \bbd_2}^{\bbd_2}$ that connect $\partial \bbd_1$, $\bbd_1$ and $\partial \bbd_2$, $\bbd_2$, respectively, and are collision-free with respect to the other destinations. The area of the shortest path $\bbp_{\partial \bbd_1}^{\partial \bbd_2}$ that connects the boundary points $\partial \bbd_1$ and $\partial \bbd_2$ is bounded as
\begin{align}\label{eq:proofThm33}
	|\bbp_{\partial \bbd_1}^{\partial \bbd_2}| 
	\le 2 d_{\max}(\ccalD_{A_2}, \ccalD_{A_1})\hat{r}
\end{align}
where $d_{\max}(\ccalD_{A_2}, \ccalD_{A_1})$ is the distance between $\ccalD_{A_1}$, $\ccalD_{A_2}$ and $d_{A_2, \ccalD} = d_{\max}(\ccalD_{A_2}, \ccalD_{A_1})$ by definition. From the condition \eqref{eq:prioritizedOfflineCompletenessCondition}, we have
\begin{align}\label{eq:proofThm34}
	|\ccalE \!\setminus\! (\Delta \!\cup\! \ccalS \!\cup\! \ccalD)| \!\ge\! 2 d_{A_1} \hat{r} + 2 (d_{A_2,\ccalS}\!+\!d_{A_2,\ccalD})\hat{r}.
\end{align}	
Thus, we can optimize $\Delta^{(1)}$ to $\Delta^{(2)}$ such that the paths $\bbp_{\partial \bbs_2}^{\partial \bbs_1}$ and $\bbp_{\partial \bbd_1}^{\partial \bbd_2}$ are obstacle-free. Following the proof of Theorem \ref{thm:offlineCompleteness}, we can establish the paths $\bbp_{\bbs_2}^{\bbs_1}$ and $\bbp_{\bbd_1}^{\bbd_2}$ that connect $\bbs_2$, $\bbs_1$ and $\bbd_1$, $\bbd_2$, respectively, and are collision-free w.r.t. the other initial positions, destinations and the optimized obstacle region $\Delta^{(2)}$. By concatenating $\bbp_{\bbs_2}^{\bbs_1}$, $\bbp_1$ and $\bbp_{\bbd_1}^{\bbd_2}$, we can establish the collision-free path $\bbp_2$ for agent $A_2$. 

Lastly, we follow the above procedure to optimize the obstacle region $\Delta^{(2)}$ to $\Delta^{(n)}$ for agents $A_3, \ldots, A_n$, recursively, and establish the paths $\{\bbp_i\}_{i=1}^n$ that are collision-free w.r.t. the other initial positions, destinations and the optimized obstacle region. This shows the optimized environment is ``well-formed'' for all agents $\{A_i\}_{i=1}^n$ and completes the proof by using Assumption \ref{as:trajectory} and Theorem 4 in \cite{vcap2015complete}.


\section{Proof of Corollary \ref{coro:incompletePrioritizedOfflineEO}}\label{proof:Corollary1}

Consider a sub-system of $b$ agents $\{A_i\}_{i=1}^b$ with highest priorities. Following the proof of Theorem \ref{thm:prioritizedOfflineCompleteness}, we can optimize the obstacle region $\Delta$ to $\Delta^*$ such that for the initial position $\bbs_i$ and destination $\bbd_i$ of agent $A_i$, there exists a collision-free path $\bbp_i$ that connects $\bbs_i$ and $\bbd_i$ for $i=1,\ldots,b$. Therefore, $\Delta^*$ is ``well-formed'' w.r.t. the considered sub-system and the navigation tasks of $b$ agents $\{A_i\}_{i=1}^b$ can be carried out successfully without collision. 

For the rest of the agents $A_j$ for $j=b+1,\ldots,n$, assume that the initial position $\bbs_j$ of agent $A_j$ is within the same starting component as $\bbs_i$ of agent $A_i$ with $i\in \{1,\ldots,b\}$, i.e., $\ccalS_{A_j} = \ccalS_{A_i}$. Denote by $\partial \bbs_i$ the boundary point where the path $\bbp_i$ intersects with the boundary of the starting component $\ccalS_{A_i}$. Since the initial positions $\{\bbs_i\}_{i=1}^n$ in $\ccalS$ are distributed in a way such that $d(\bbs_{i_1},\bbs_{i_2}) \ge 2 \hat{r}$ and $d(\bbs_i,\partial \ccalS) \ge 2 \hat{r}$ from Assumption \ref{as:initialDistribution} and $\bbs_j$ is in the same starting component as $\bbs_i$, there exists a path $\bbp_{\bbs_j}^{\partial \bbs_i}$ in $\ccalS_{A_i}$ that connects $\bbs_j$, $\partial \bbs_i$ and is collision-free w.r.t. the other initial positions. Similar result applies to the destinations $\bbd_j$ and $\bbd_i$. That is, if $\bbd_j$ is within the same destination component as $\bbd_i$, i.e., $\ccalD_{A_j} = \ccalD_{A_i}$, there exists a path $\bbp_{\bbd_j}^{\partial \bbd_i}$ in $\ccalD_{A_i}$ that connects $\bbd_j$, $\partial \bbd_i$ and is collision-free w.r.t. the other destinations. Since the boundary points $\partial \bbs_i$ and $\partial \bbd_i$ can be connected by the collision-free path $\bbp_i$, we can establish the collision-free path $\bbp_j$ that connects $\bbs_j$ and $\bbd_j$ by concatenating $\bbp_{\bbs_j}^{\partial \bbs_i}$, $\bbp_i$ and $\bbp_{\bbd_j}^{\partial \bbd_i}$. Therefore, the optimized obstacle region $\Delta^*$ is ``well-formed'' w.r.t. agent $A_j$ as well, the navigation task of which can be carried out successfully without collision. The same result holds for all agents $A_j \in \{A_{b+1}, \ldots, A_n\}$ satisfying the above conditions, completing the proof.

\section{Proof of Theorem \ref{thm:prioritizedOnlineCompleteness}}\label{proof:Theorem4}

We prove the theorem following Theorem \ref{thm:onlineCompleteness}. First, we re-formulate the navigation task as $H$ sub-navigation tasks for each agent. Then, we optimize the obstacle region to guarantee the completeness of sub-navigation tasks successively. Lastly, we show the completeness of the entire navigation by concatenating these sub-navigation tasks and complete the proof by limiting $H \to \infty$. 

Let $T$ be the maximal operation time required by trajectories $\{\bbp_i(t)\}_{i=1}^n$ with velocities $\{\bbv_i(t)\}_{i=1}^n$ and $\big\{\bbp_i(hT/H)\big\}_{h=0}^{H}$ be the intermediate positions with $\bbp_i(0)=\bbs_i$ and $\bbp_i(T)=\bbd_i$ for $i=1,...,n$. The goal of the $h$th sub-navigation task for agent $A_i$ is from $\bbp_i((h-1)T/H)$ to $\bbp_i(hT/H)$ and the operation time required by each sub-navigation task is $\delta t = T/H$. In this context, we separate the procedure of online environment optimization as a number of time slices with duration $2\delta t$. At each time slice, we change the obstacle region for sub-navigation tasks with duration $\delta t$ and navigate the agents with duration $\delta t$ in an alternative manner. From the condition \eqref{eq:prioritizedOnlineCompletenessCondition}, the area of the obstacle region that can be changed at each time slice is
\begin{align}\label{eq:proofthm41}
	2 b \hat{r} \|\hat{\bbv}\|_2 \delta t \!\le \!\!\frac{\big|(\!\Delta \!\bigcup\! \widetilde{\Delta}) \!\setminus\! (\Delta \bigcap \widetilde{\Delta})\big|}{2} \!\! <\! 2 (b\!+\!1) \hat{r} \|\hat{\bbv}\|_2 \delta t
\end{align}
where $\Delta$ is the original obstacle region and $\widetilde{\Delta}$ is the changed obstacle region. From \eqref{eq:proofthm41} and the proof of Theorem \ref{thm:onlineCompleteness}, at each time slice, we can change the obstacle region to make the environment ``well-formed'' w.r.t. the sub-navigation tasks of only $b$ agents instead of all $n$ agents, i.e., it only guarantees the success of $b$ sub-navigation tasks. For the agents whose sub-navigation tasks are not selected for environment optimization, we keep them static until the next time slice when their sub-navigation tasks are selected. By tuning the selections of $b$ sub-navigation tasks across time slices, we prove there exists a selection scheme such that sub-navigation tasks of all agents can be carried out successfully without collision and the navigation time of the agents can be bounded inverse proportionally to their priorities. 

Specifically, let $T_1 = H_1 (2 \delta t)$ with $H_1 \ge H$ be the maximal navigation time required by agent $A_1$ with the highest priority $\rho_1$, i.e., agent $A_1$ requires completing its $H$ sub-navigation tasks within $H_1$ time slices, and $T_i = H_i  (2 \delta t) = \lceil H_1 \rho_1 / \rho_i \rceil  (2 \delta t)$ be the maximal navigation time required by agent $A_i$ with the priority $\rho_i$ for $i=2,\ldots,n$. For agent $A_1$, there exists a scheme that selects the sub-navigation task of $A_1$ for environment optimization $H$ times within $H_1$ time slices because $H_1 \ge H$ and thus, its $H$ sub-navigation tasks can be carried out successfully without collision. By concatenating these sub-tasks, the navigation task of $A_1$ can be carried out successfully without collision within the required time.

For agents $A_1$ and $A_2$, if it holds that 
\begin{align}\label{eq:proofthm415}
	\min\{1, b\} (H_2-H_1) + \min\{2, b\} H_1 \ge 2 H
\end{align}
there exists a scheme that selects the sub-navigation tasks of $A_1$ and $A_2$ for environment optimization $H$ times during $H_2$ time slices, respectively. The minimal operations $\min\{1, b\}$ and $\min\{2, b\}$ in \eqref{eq:proofthm415} represent the facts that single agent can be selected at most once at each time slice, and each time slice can select at most $b$ agents. Thus, the navigation tasks of $A_1$ and $A_2$ can be carried out successfully without collision within the required time. Analogously for agents $\{A_1,\ldots,A_i\}$, if it holds that 
\begin{align}\label{eq:proofthm42}
	\sum_{j=1}^{i} \min\{j, b\} (H_{i+1-j} - H_{i-j}) \ge i H
\end{align}
with $H_0 = 0$ by default, there exists a scheme that selects the sub-navigation tasks of $\{A_1,\ldots,A_i\}$ for environment optimization $H$ times during $H_i$ time slices, and their navigation tasks can be carried out successfully without collision within the required time. Therefore, we conclude that if
\begin{align}\label{eq:proofthm43}
	\sum_{j=1}^{n} \min\{j, b\} (H_{n+1-j} - H_{n-j}) \ge n H
\end{align}
there exists a scheme that selects the sub-navigation tasks of all $n$ agents for environment optimization $H$ times during $H_n$ time slices, and all navigation tasks can be carried out successfully without collision within the required time. Since $H_i = \lceil H_1 \rho_1 / \rho_2 \rceil$, i.e., $H_i$ can be represented by $H_1$ and the associated priorities for $i=1,\ldots,n$, we can rewrite \eqref{eq:proofthm43} as
\begin{align}\label{eq:proofthm44}
	\!\sum_{j=1}^{n}\! \min\{j,\! b\} \Big(\Big\lceil \frac{H_1 \rho_1}{\rho_{n+1-j}} \Big\rceil \!-\! \Big\lceil \frac{H_1 \rho_1}{\rho_{n-j}}  \Big\rceil\Big) \!\ge\! n H.
\end{align}
Since $\{\rho_i\}_{i=1}^n$, $H$ are given and the left-hand side of \eqref{eq:proofthm44} increases with $H_1$, there exists a large enough $H_1$ such that \eqref{eq:proofthm44} holds. Therefore, the navigation tasks of all agents can be carried out successfully and the navigation time of agent $A_i$ is bounded by $T_i = H_i (2 \delta t) = \lceil H_1 \rho_1 / \rho_i \rceil (2 \delta t)$. 

When $H$ is sufficiently large, i.e., $\delta t$ is sufficiently small, the agents and the obstacles can be considered moving simultaneously. Since $H_1 \ge H$ is sufficiently large as well, we have $ \lceil H_1 \rho_1 / \rho_i \rceil (2\delta t) \approx  2 H_1 \rho_1 \delta t / \rho_i$. With this observation and the conclusion obtained from \eqref{eq:proofthm44}, we complete the proof, i.e., the navigation tasks of all agents can be carried out successfully without collision and the navigation time of agent $A_i$ is bounded by $T_i = C_T / \rho_i$ for $i=1,\ldots,n$ where $C_T = 2 H_1 \rho_1 \delta t$ is the time constant.

\section{Proof of Corollary \ref{coro:IncompletePrioritizedOnlineEO}}\label{proof:Corollary2}

From the proof of Theorem \ref{thm:prioritizedOnlineCompleteness}, for any integer $1 \le \eta \le n$, if it holds that
\begin{align}\label{eq:proofcoro21}
	\sum_{j=1}^{\eta}\! \min\{j, b\} \Big(\Big\lceil\! \frac{H_1 \rho_1}{\rho_{\eta+1-j}} \Big\rceil \!-\! \Big\lceil \frac{ H_1 \rho_1}{\rho_{\eta-j}} \Big\rceil\Big) \!\ge\! \eta H
\end{align}
there exists a scheme with a sufficiently large $H_1$ that selects the sub-navigation tasks of $\eta$ agents for environment optimization $H$ times during $H_{\eta} = \lceil H_1 \rho_1 / \rho_{\eta+1-j} \rceil$ time slices and the navigation tasks of agents $\{A_i\}_{i=1}^{\eta}$ can be carried out successfully without collision within time $\{T_i\}_{i=1}^\eta$, respectively. Since the agents are required to reach their destinations within time $T_{\max}$, i.e., the maximal navigation time is $T_{\max}$, the maximal (allowed) number of time slices is $H_{\max} = \lfloor T_{\max} / 2 \delta t \rfloor$. This is equivalent to requiring
\begin{align}\label{eq:proofcoro22}
	H_{\eta} = \Big\lceil \frac{H_1 \rho_1}{\rho_{\eta}}  \Big\rceil \le \Big\lfloor \frac{T_{\max}}{2 \delta t} \Big\rfloor
\end{align}
because $H_{\eta}$ is the maximal number of time slices required by the first $\eta$ agents. Thus, we have
\begin{align}\label{eq:proofcoro23}
	H_1 \le \frac{\big\lfloor \frac{T_{\max}}{2 \delta t} \big\rfloor \rho_\eta}{\rho_1} 
\end{align}
Since the left-hand side of \eqref{eq:proofcoro21} increases with $H_1$, substituting \eqref{eq:proofcoro23} into \eqref{eq:proofcoro21} yields
\begin{align}\label{eq:proofcoro24}
	&\sum_{j=1}^{\eta}\! \min\{j, b\} \Big(\Big\lceil\! \frac{H_1 \rho_1}{\rho_{\eta+1-j}} \Big\rceil \!-\! \Big\lceil \frac{ H_1 \rho_1}{\rho_{\eta-j}} \Big\rceil\Big) \\
	&\le \sum_{j=1}^{\eta}\! \min\{j, b\} \Big(\Big\lceil \frac{\big\lfloor \frac{T_{\max}}{2 \delta t} \big\rfloor \rho_\eta}{\rho_{\eta+1-j}} \Big\rceil \!-\! \Big\lceil \frac{\big\lfloor \frac{T_{\max}}{2 \delta t} \big\rfloor \rho_\eta}{\rho_{\eta-j}} \Big\rceil\Big). \nonumber 
\end{align}
By substituting \eqref{eq:proofcoro24} into \eqref{eq:proofcoro21}, we have
\begin{align}\label{eq:proofcoro25}
	\!\sum_{j=1}^{\eta}\! \min\{j, b\} \Big(\Big\lceil \frac{\big\lfloor \frac{T_{\max}}{2 \delta t} \big\rfloor \rho_\eta}{\rho_{\eta+1-j}} \Big\rceil \!-\! \Big\lceil \frac{\big\lfloor \frac{T_{\max}}{2 \delta t} \big\rfloor \rho_\eta}{\rho_{\eta-j}} \Big\rceil\Big) \!\ge\! \eta H.
\end{align}
By using the fact $T = H \delta t$ where $T$ is the maximal operation time of trajectories $\{\bbp_i(t)\}_{i=1}^n$, we get
\begin{align}\label{eq:proofcoro26}
	\!\sum_{j=1}^{\eta}\! \min\{j, b\} \Big(\Big\lceil \frac{\big\lfloor \frac{HT_{\max}}{2 T} \big\rfloor \rho_\eta}{\rho_{\eta+1-j}} \Big\rceil \!-\! \Big\lceil \frac{\big\lfloor \frac{H T_{\max}}{2 T} \big\rfloor \rho_\eta}{\rho_{\eta-j}} \Big\rceil\Big) \!\ge\! \eta H.
\end{align}
When $H$ is sufficiently large, we have $\lceil \lfloor HT_{\max}/(2T) \rfloor \rho_\eta/\rho_{\eta+1-j}\rceil \approx HT_{\max}\rho_\eta/(2 T \rho_{\eta+1-j})$ and $\lceil \lfloor HT_{\max}/(2T) \rfloor \rho_\eta/\rho_{\eta-j}\rceil \approx HT_{\max}\rho_\eta/(2 T \rho_{\eta-j})$, and \eqref{eq:proofcoro26} is equivalent as 
\begin{align}\label{eq:proofcoro27}
	\!\sum_{j=1}^{\eta}\! \min\{j, b\} \Big( \frac{T_{\max} \rho_\eta}{2 T \rho_{\eta+1-j}} \!-\! \frac{T_{\max} \rho_\eta}{2 T \rho_{\eta-j}} \Big) \!\ge\! \eta.
\end{align}
By setting $n_p$ as the maximal $\eta$ that satisfies \eqref{eq:proofcoro27} and following the proof of Theorem \ref{thm:prioritizedOnlineCompleteness}, the navigation tasks of $\{A_i\}_{i=1}^{n_p}$ can be carried out successfully without collision within the required time $T_{\max}$, completing the proof.

\section{Proof of Theorem \ref{thm:safetyGuarantee}}\label{proof:Theorem5}

For a feasible solution of problem \eqref{eq:alternativeLearningProblem} with the constraint constant $\ccalC_q = (1-\delta + \eps)/(1-\gamma)$ [cf. \eqref{eq:constraintConstants}], it satisfies that
\begin{align}\label{eq:thm51}
	&\mathbb{E}\Big[\!\sum_{t=1}^\infty\!\! \gamma^t\! \mathbbm{1}[h_q(\bbX_o^{(t)}\!,\! \bbU_o^{(t)}) \!\le\! 0]\Big] \!=\! \sum_{t=1}^\infty\! \gamma^t \mathbbm{P}[h_q(\bbX_o^{(t)}, \bbU_o^{(t)}) \le 0] \nonumber \\
	& \ge \frac{1-\delta + \eps}{1-\gamma} = \frac{1}{1-\gamma} - \frac{\delta - \epsilon}{1-\gamma},~\for~q=1,\ldots,Q 
\end{align}
where the linearity of the expectation and \eqref{eq:probability} are used. For the term $\sum_{t=1}^\infty \gamma^t \mathbbm{P}[h_q(\bbX_o^{(t)}, \bbU_o^{(t)}) \le 0]$, we have 
\begin{align}\label{eq:thm52}
	&\sum_{t=1}^\infty\!\! \gamma^t \mathbbm{P}[h_q(\bbX_o^{(t)}\!,\! \bbU_o^{(t)}\!) \!\le\! 0] \!=\!\! \sum_{t=1}^\infty\!\! \gamma^t \!\!-\!\! \sum_{t=1}^\infty\!\! \gamma^t \mathbbm{P}[h_q(\bbX_o^{(t)}\!\!,\! \bbU_o^{(t)}\!) \!>\! 0] \nonumber \\
	&= \frac{1}{1-\gamma} - \sum_{t=1}^\infty \gamma^t \mathbbm{P}[h_q(\bbX_o^{(t)}, \bbU_o^{(t)}) > 0]. 
\end{align}
By comparing \eqref{eq:thm51} and \eqref{eq:thm52}, we get
\begin{align}\label{eq:thm53}
	&\sum_{t=1}^\infty \gamma^t \mathbbm{P}[h_q(\bbX_o^{(t)}, \bbU_o^{(t)}) > 0] \le \frac{\delta - \epsilon}{1 - \gamma}.
\end{align}
For the maximal time horizon $T$, we have
\begin{align}\label{eq:thm54}
	&\sum_{t=1}^T \gamma^t \mathbbm{P}[h_q(\bbX_o^{(t)}, \bbU_o^{(t)}) > 0] \\
	&\le \sum_{t=1}^\infty \gamma^t \mathbbm{P}[h_q(\bbX_o^{(t)}, \bbU_o^{(t)}) > 0] \le \frac{\delta - \epsilon}{1 - \gamma} \nonumber
\end{align}
because each term in the summation is non-negative. Then, by setting $\epsilon = \delta (1-\gamma^T (1 - \gamma)) < \delta$ and substituting the latter into \eqref{eq:thm54}, we get 
$\sum_{t=1}^T \gamma^t \mathbbm{P}[h_q(\bbX_o^{(t)}, \bbU_o^{(t)}) > 0] \le \gamma^T \delta$. 
Since $\gamma^T \le \gamma^t$ for all $t \le T$, we have 
$\sum_{t=1}^T \gamma^T \mathbbm{P}[h_q(\bbX_o^{(t)}, \bbU_o^{(t)}) > 0] \le \gamma^T \delta$ 
and thus
\begin{align}\label{eq:thm57}
	&\sum_{t=1}^T \mathbbm{P}[h_q(\bbX_o^{(t)}, \bbU_o^{(t)}) > 0] \le \delta.
\end{align}
By leveraging the Boole-Frechet-Bonferroni inequality, we complete the proof
\begin{align}\label{eq:thm58}
	\mathbb{P}\big[\cap_{0 \le \tau \le t} \{h_q(\bbX_o^{(\tau)}, \bbU_o^{(\tau)}) \le 0\}\big] \ge 1-\delta.
\end{align}

\bibliographystyle{IEEEtran}
\bibliography{myIEEEabrv,biblioOp,AP_bib}

\end{document}